\documentclass[12pt]{amsart}
\usepackage{graphicx}
\usepackage{geometry}
\usepackage{amssymb}
\usepackage{setspace}
\usepackage[dvipsnames]{pstricks}
\usepackage{pst-node}
\usepackage{pst-slpe}
\usepackage{pst-tree}
\usepackage{hyperref}
\hypersetup{
    colorlinks=false,
    pdfborder={0 0 0},
}
\usepackage{amsfonts}
\usepackage{amsmath}
\usepackage{amssymb}
\usepackage{graphicx}
\usepackage{appendix}
\usepackage{multirow}
\usepackage{ulem}
\usepackage{setspace}
\usepackage[dvipsnames]{pstricks}
\usepackage{pst-node}
\usepackage{pst-slpe}
\usepackage[square,comma]{natbib}
\usepackage{amsthm}
\usepackage{lmodern}
\usepackage[T1]{fontenc}
\usepackage{mathtools}

\usepackage{amssymb}
\usepackage{pifont}
\newcommand{\cmark}{\ding{51}}%
\newcommand{\xmark}{\ding{55}}%

\newtheorem*{ax1}{Axiom}

\usepackage{natbib}

  \theoremstyle{definition}
  \newtheorem{defn}{\protect\definitionname}
  \theoremstyle{plain}
  
  \theoremstyle{plain}
  \newtheorem{claim}{Claim}
  \theoremstyle{definition}
  \newtheorem{ax}{\protect\axiomname}
  \theoremstyle{plain}
  \newtheorem{thm}{\protect\theoremname}
  \theoremstyle{plain}
  \newtheorem{lemma}{Lemma}
  \theoremstyle{plain}
  \newtheorem{cor}{\protect\corollaryname}
  \theoremstyle{plain}
  \newtheorem{proposition}{\protect\propositionname}
\theoremstyle{remark}

  \theoremstyle{definition}
  \newtheorem{example}{Example}
    
  \newtheorem*{ass*}{Assumption}
  
  \providecommand{\axiomname}{Axiom}
  \providecommand{\conjecturename}{Conjecture}
  \providecommand{\definitionname}{Definition}
  \providecommand{\corollaryname}{Corollary}
  \providecommand{\theoremname}{Theorem}
\providecommand{\propositionname}{Proposition}

\begin{document}
\title[]{Choice with Endogenous Categorization$^\ast$}
\author[]{Andrew Ellis$^\dag$}
\author[]{Yusufcan Masatlioglu$^\S$}
\thanks{$^{\ast }$We thank  David Dillenberger, Erik Eyster, Nicola Gennaioli, Matt Levy, Collin Raymond, the anonymous referees, Andrei Shleifer, Kate Smith, Rani Spiegler, Tomasz Strzalecki as well as the conference/seminar participants at BRIC 2017, CETC 2017, SAET 2017, Lisbon Meetings 2017, ESSET 2019, Brown, UPenn, Pompeu Fabra, and Harvard for helpful comments and discussions. This project began at ESSET Gerzensee, whose hospitality is gratefully acknowledged.}
\thanks{$^\dag$ Department of Economics, London School of Economics, Haughton Street, London, WC2A 2AE. Email: \texttt{a.ellis@lse.ac.uk.}}
\thanks{$^\S$ University of Maryland, 3147E Tydings Hall, 7343 Preinkert Dr.,  College Park, MD 20742. E-mail: \texttt{yusufcan@umd.edu}}
\date{1 February, 2021}
\begin{abstract} We propose and axiomatize the categorical thinking model (CTM) in which the framing of the
decision problem affects how agents categorize alternatives, that in turn affects their evaluation of it. Prominent models of salience, status quo bias, loss-aversion, inequality aversion, and present bias all fit under the umbrella of CTM. This suggests categorization is an underlying mechanism of key departures from the neoclassical model of choice. 
We specialize CTM to provide
a behavioral foundation for the salient thinking model of  \cite{bordalo2013salience} that highlights its strong predictions and distinctions from other models.
\end{abstract}
\maketitle
\newpage
\begin{quote}
``\textit{All organisms assign objects and events in the environment to separate classes or categories...  Any species lacking this ability would quickly become extinct.}''
--\cite{ashby2005human}
\end{quote}
\section{Introduction}

Categories shape how people perceive and react to the world. A real estate agent shows  clients a house in a worse neighborhood before showing them the one the agent wants to sell, so that they categorize the target's neighborhood as a gain rather than a loss. A worker may reject a higher paying job offer in a different city because the worker does not categorize  it as unambiguously better than the status quo.
A fan categorizes a \$5 soda as a bargain at their favorite team's home stadium but a rip-off in a grocery store.  
A negotiator rejects, and refuses to make, offers categorized as unfair. 
An experimental subject is  willing to wait an extra week to turn a reward of \$100 into \$110 only when both rewards are categorized as long term. This paper develops a model that generates these behaviors through a common underlying mechanism: categorization. 

We propose and axiomatize the Categorical Thinking Model (CTM) based on two key features of these examples:  categorization depends on the context and affects how the decision maker (DM) evaluates alternatives.   A DM  conforming to CTM first groups objects into categories based on a reference point, and then maximizes a category- and reference-dependent utility function over alternatives. While the reference point does not affect her ranking within a given category, it may affect comparisons across categories.
We show that a number of important models across different choice environments are special cases of CTM. Our analysis suggests categorization as their common cognitive underpinning and reveals their common behavior as well as the behavior that distinguishes  them.

Prominent models of  loss-aversion, status quo bias, salience, inequality aversion, present bias, and others all fit under the umbrella of CTM.
A loss-averse DM \citep{TK91} categorizes alternatives according to which attributes are
gains and which are losses, then treats the two very differently.
A DM subject to status quo bias \citep{MO05} categorizes alternatives according to whether they unambiguously improve on the status quo, then penalizes the ones that do not.
A salient-thinking DM \citep{bordalo2013salience} categorizes alternatives according to which attribute stands out most, then  overweights that attribute.
An inequality-averse DM \citep{fehrschmidt99} categorizes social allocations according to whether she feels  envy  or guilt towards each of the others, then evaluates the allocation accordingly.
A quasi-hyperbolic DM \citep{phelpspollak1968} categorizes dated rewards as short- or long-term, then discounts the former at a higher rate.

We show that a family of reference-dependent preference relations conforms to CTM if, and only if, the ranking between alternatives belonging to the same category  does not depend on the reference point   and satisfies some standard axioms.  Within CTM, a large latitude of models is allowed, yet they share the important commonalities identified by our result.
For instance, the salient-thinking model \citep{bordalo2013salience} (BGS) and the constant loss-aversion model \citep{TK91} (TK) both fall under CTM, so our result establishes their common foundations and that categorization can serve as their common  psychological underpinning. 
The result  formalizes the behavioral implications separating models in CTM from those outside it. 
For example,  the context shifts weight between attributes in both the focus-weighted utility model \citep{koszegi2013model} and BGS.   The former necessarily violates  reference irrelevance and so is not a special case of CTM, while BGS satisfies it, as well as the other axioms.

While we initially consider exogenously-specified categories,
our framework has the advantage of allowing categories to be derived endogenously from choice behavior.%
\footnote{Non-choice data is an additional source of identification and can be used in conjunction with or in lieu of our methods.}
We provide a method to identify the categorization based on the changes in trade-offs between attributes.
This allows our model to consider phenomena for which the psychology makes only partial category predictions, like salience.
By endogenously identifying categories, the result extends the model's applicability beyond cases with unambiguous categorization,  such as gains and losses.

In economics, the most prominent model of salience, BGS,  accounts for a number of empirical anomalies, but because its new components are unobservable, understanding all of its implications for behavior can be difficult. 
We apply our results to provide  the first complete characterization of the observable choice behavior equivalent to BGS, clarifying and identifying the nature of the assumptions used by it.\footnote{In a recent paper, \cite{lanzani2019} provides a characterization of the related model for risk. This paper provides complimentary insights on the role of salience in different environments. For instance, the only objects in both environments are sure-monetary payments on which both predict  $x \succ y \iff x>y$. }
The first crucial step towards understanding the model is understanding its novel \emph{salience function} that determines which attribute stands out for a given reference point. 
While the salience function influences which attribute is salient, the weight given to each attribute is independent of its magnitude, so BGS is a special case of CTM. We identify properties that its categories must satisfy and the regularities that distinguish it from other CTMs.

In some models,  the set of available options endogenously determines the reference point. For instance, the reference point is the average of the available alternatives in BGS, so varying the budget set affects the salience of, and so the DM's evaluation of, a given alternative. 
Our final contribution addresses this challenge by extending our characterization of CTM and the identification of categories to accommodate an endogenous reference point. We take a choice correspondence describing the DM's behavior, and assume each menu is mapped to a reference point, such as the average alternative in it.  
As long as the reference point varies systematically with the choice problem, we characterize the properties of a choice correspondence that conforms to CTM.  Specifically, we show that if the DM's choices obey the natural analogs of our axioms in the exogenous-reference setting, then CTM rationalizes her behavior.
Together, the results admit a characterization of BGS with both an endogenous reference point and endogenous identification of categories, unlike our previous results that relied on exogeneity of at least one of the two.

The paper proceeds as follows. The next subsection provides a brief overview of the relevant psychology literature on categorization. In Section \ref{sec: Model}, we introduce the CTM and discuss the models covered under its umbrella. In Section \ref{sec: axiomatization}, we axiomatize the CTM and  compare and contrast  the models of riskless choice discussed in Section \ref{sec: Model}.  Section \ref{sec: BGS} contains our analysis of the salient thinking model. In Section \ref{sec: Choice Correspondence} we introduce the endogenous reference point setting and apply our axiomatizations of the CTM to it.   Section \ref{sec: conclusion} concludes by comparing our results with the related economics literature.
\subsection{Psychology of Categorization}
There is  a long literature in psychology and marketing discussing categorization. Recent review articles include \cite{ashby2005human, loken2006consumer, loken2008categorization}, and \cite{cosmides2013evolutionary}. Much of the literature focuses on how subjects form categories and how they add new alternatives to existing categories. The two properties on which CTM is based are well-documented by this literature. 

First, categories are context dependent. \cite{tversky1977features, tversky1978studies} 
	present evidence that replacing one item in a set of objects can
drastically alter how people categorize the  remaining objects.   \cite{tversky1978studies} argue that categorization {\it{``is generally not invariant with respect to changes in context or frame of reference.''}} For example, they show that subjects put East Germany and West Germany into the same category when the salient feature is geography or cultural background, but categorize the two differently when politics are salient. 
Similarly, \cite{Choi2016} posit that depending on the context, a person may categorize an Apple Watch as a tech product, a fashion product, a fitness product, or just as a watch.
 \cite{ratneshwar1991substitution}
 show that subjects categorize ice cream and cookies together in terms of similarity (e.g., they are both desserts), but categorize ice cream and hot dogs together in terms of usage (e.g., both are good snacks to have at the pool).
\cite{stewart2002sequence} present  evidence that information about the relative magnitude of sounds that is derived from a comparison with a  reference point is used to categorize them.

Second, how a person categorizes an object affects its final valuation.
In a classic series of experiments, \cite{Rosch1975} shows that subjects perceptually encode differently categorized but physically identical stimuli as distinct objects.
\cite{wankeetal1999lobster} demonstrate that people evaluate wine more positively when it is in the same category as lobster instead of with cigarettes.
\cite{mogilner2008mere} show that categorizing goods differently results in varying levels of satisfaction. \cite{chernev2011dieter} shows that bundling a healthy food item with a junk food item reduces the reported caloric content beyond that of the junk food alone. 

Moreover, CTM models categories as regions in the alternative space.  This  closely tracks psychology's decision bound theory. As  \citet[p. 152]{ashby2005human} describe,  the subject ``\textit{partition[s] the stimulus space into response regions... determines which region the percept is in, and then emits the associated response.}'' \cite{ashbygott1988} show it can accommodate examples incompatible with other theories of category formation, such as prototype theory. Moreover, there is substantial experimental support for it, such as \cite{AshbyWaldron1999,anderson1991adaptive, love2004sustain}.

\section{Model}
\label{sec: Model}
The DM makes a choice of an  alternative in $X=\mathbb{R}^n_{++}$, and we focus on $n=2$ when not otherwise noted.%
\footnote{\label{fn: Rn}We note when there is a distinction between general $n$ and $n=2$. Theorem \ref{thm: BGS regions} and the results that rely on it use the full structure of $\mathbb{R}^2_{++}$. 	
	The remaining results all generalize to any $X$ that is a finite Cartesian product of open, linearly ordered, separable, connected sets endowed with the order topology, where $X$ itself has the product topology.}
The next subsections explore three different interpretations of $X$ in different contexts: as a riskless object with different attributes, as a dated reward or consumption stream, and as an allocation of consumption across individuals.
We often use the convention of writing $x\in X$ as $(x_i, x_{-i})$ with $x_{-i}$ denoting the components of $x$ different from $i$.

The DM maximizes a complete and transitive preference relation $\succsim_r $ over $X$ 
when her reference point is $r \in X$.  As usual, $\succ_r$ denotes strict preference  and $\sim_r$ indifference.
In Sections \ref{sec: Model}-\ref{sec: BGS}, we assume that the reference point is exogenously given, so our primitive is a family $\{ \succsim_r\}_{r\in X}$. This isolates the effects of categorization from that of reference point formation and allows easier comparison with the existing literature like \cite{TK91}. We relax this assumption in Section \ref{sec: Choice Correspondence} to allow endogenous reference point formation.
\subsection{Categorical Thinking Model}
We model category formation via a function that maps reference points to subsets of alternatives belonging to each category. We allow categories to have a very general structure. 
\begin{defn}\label{defn: category function} 
	A vector-valued function $\mathcal{K} =(K^1, K^2, \dots, K^m)$ is a {\it{category function}} if each $K^k: X  \rightarrow {2^{X}}$ satisfies the following properties:
	\vspace{-0.5cm}
	\begin{enumerate}
		\item $K^k(r)$ is a non-empty, regular open set, and $cl(K^k(r))$ is connected,\footnote{A set $A$ is regular open if $A=int(cl(A))$.}
		\item $\bigcup_{k=1}^m K^k (r)$ is dense,
		\item $K^k(r) \bigcap K^l (r) = \emptyset$ for all $k\neq l$, and
		\item $K^k(\cdot)$ is continuous.\footnote{That is, each $K^k$ is both upper and lower hemicontinuous when viewed as a correspondence.}
	\end{enumerate}
\end{defn}
We interpret the properties of the category function as follows. Every category contains some alternative for every
reference point. If a particular product, say $x$, belongs to the category $k$, then so do all products that are close enough to $x$. For any two points in the same category, we can find a path in its closure, so categories cannot be the union of ``islands.'' Almost every alternative is in at least one category, and none are in two categories. Further, if the reference point does not change too much, then neither do the categories.

Categories arise from the psychology of the phenomenon to be modeled. For CTM to be applicable, one must know or infer the category function. Often the psychology makes unambiguous predictions about categorization. For instance, with gain-loss utility, a DM  treats alternatives that dominate the reference point differently than those better in only one dimension. 
Other times, non-choice data such as hypothetical questions, subjective valuations, reaction times, physiological reactions, and neurological responses combine with the psychology to make unambiguous predictions.
When only partial predictions are possible even after adjusting for other sources of information, the modeler must infer the categorization  (see Propositions \ref{prop: general category ID},  \ref{prop: revealing regions}, and 
\ref{prop: revealing regions with c}).

Say that a function $ U^k : X \rightarrow \mathbb{R}$ is additively separable and monotonic if $U^k(x) = \sum_{i=1}^n U^k_i(x_i) $ where each $U^k_i(\cdot)$ is strictly monotone and continuous.%
\footnote{That is, each $U^k_i$ either strictly increases or strictly decreases.}
We can now state the formal representation.
\begin{defn} 
	The family $\{\succsim_r\}_{r \in X}$ conforms to the \emph{Categorical Thinking Model (CTM)} under category function $\mathcal{K} =(K^1, K^2, \dots, K^m)$  if for each category $k$ there is an additively separable and monotonic  $ U^k: X \rightarrow \mathbb{R}$  
	and a family $\{U^k(\cdot|r)\}_{r\in X}$ of continuous, increasing transformations of $U^k(\cdot)$ so that for any $r \in X$   \[
	\text{if }x \in K^k(r)\text{ and }y \in K^l(r)\text{, then
	 }x \succsim_r y \iff  U^k(x|r)  \geq U^l(y|r).
	\]
\end{defn}

A DM conforming to CTM values each alternative  in a way that depends not only on its attributes  but also on the category to which it belongs.
She values $x$ at $U^k(x|r)$ when $x$ is categorized as $k$ for $r$, and since $U^k(\cdot|r)$ typically does not equal $U^l(\cdot|r)$, her categorization affects her valuation.%
\footnote{When $U^k(x|r)\neq U^l(x|r)$ for some $x\in X$, discontinuities may occur but only on the boundary between categories. This is consistent with a number of findings in the psychology literature.
As observed by \citet[p. 6]{Rosch1978},
``\textit{In the perceived world, 
information-rich bundles of perceptual and functional attributes
occur that form natural discontinuities, and ... cuts in categorization are
made at these discontinuities}.''}
On the one hand, the DM evaluates alternatives in the category  independently of the reference point because each $U^k(\cdot|r)$   is an increasing transformation of $U^k(\cdot)$ for any $r$. 
Consequently, the \emph{category utility function} $U^k(\cdot)$ governs the trade-off between attributes within category $k$. 
On the other hand, the reference point may affect the DM's choice between alternatives belong to different categories since $U^k(x|r)$ need not equal $U^k(x|r')$.

The following subclasses are of particular interest.
A DM conforms to \emph{Increasing CTM} if $U^k_i$ increases with $x_i$ for every category $k$ and dimension $i$.
She conforms to \emph{Affine CTM} if $U^k(\cdot|r)$ is an affine transformation of $U^k$ for each $r$, and to \emph{Strong CTM} if  $U^k(\cdot|r)=U^k(\cdot)$ for each $r$.

\subsubsection*{Remarks on the model }
A reference point is a specific instance of the general concept of framing. Our framework extends to cover other forms of framing, such as the intensity of advertising, the amount of light in a supermarket, and expectations in the form of lotteries (as in \cite{KR06}). Our definition of a category function extends naturally to mappings from frames to categories, and most of our results continue to hold when  behavior is described by a family of complete and transitive preferences indexed by a sufficiently well-behaved set of frames.\footnote{Specifically, a non-empty, compact, path-connected subset of a metric space. }

 Not every reference-dependent model is a CTM. For example, the general loss-aversion model of \cite{TK91} and the reference-dependent CES model of \cite{munro2003theory} do not fall into the class of CTMs since the reference point affects the marginal rate of substitution between attributes. 
Nor does it encompass all models in which the framing distorts the indifference curves:  the models of  \cite{koszegi2013model}, \cite{bhatia2013attention}, and \cite{bushong2015model} all fall outside of the CTM umbrella for the same reason.

 CTM treats categories as stark and does not allow the framing to change how the DM makes trade-offs within a category. 
 It rules out related models in which the weight on a dimension changes continuously with the reference point. Nevertheless, such models can be approximated by CTM with a large number of categories when weights depend on the position of the alternative relative to the frame, as in \cite{BGS_MAC}. In contrast, when the weighting depends on the frame alone as in \cite{koszegi2013model}, the indifference curves shift in the same way at each point. If this model could be approximated by CTM, then every category would have the same indifference curves, which would in turn imply that the frame does not affect the DM's choice.  
 
\subsection{Riskless Consumer Choice}
\label{sec: consumer choice}
In this subsection, we consider our primary application: riskless consumer  choice. To show how different models fit into our
framework, we first define psychologically relevant categories for each model and then map them to a category function. 
For the purpose of illustration,  
Figure \ref{fig:models} plots their indifference curves and categories, with darker lines indicating higher utility.
\begin{figure}[h]
	\begin{center}
		\includegraphics[width=1\textwidth]{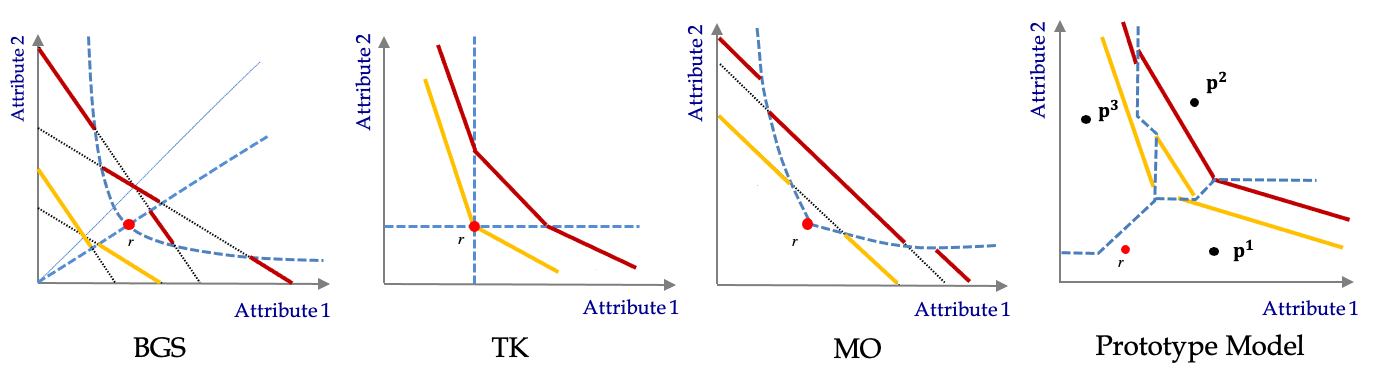}
	\end{center}
	\caption{CTM for Riskless Choice}
	\label{fig:models}
\end{figure}
\subsubsection*{Constant Loss Aversion Model (TK)}  One of the first and most broadly adopted economic insights from psychologists is that subjects treat gains and losses differently \citep{KT79}.
Accordingly, people categorize alternatives according to whether each of their attributes (or possible outcomes in the case of risk) are gains or losses. Typically, losses loom larger than gains.
\cite{TK91}  provide foundations for a reference-dependent model that captures loss aversion among riskless objects.
 
In the model, the DM determines gains and losses relative to a reference point $r$.
Given that we have two attributes, there are four different categories: (i) gain in both dimensions, (ii) loss in the first dimension and gain in the second dimension, (iii) gain in the first dimension and loss in the second dimension, and (iv) loss in both dimensions.%
\footnote{With $n$ attributes, there are $2^n$ categories.}
The \emph{gain-loss category function} $\mathcal{K}^{GL}=\left( K^{1},K^{2},K^{3},K^{4} \right)$ where $K^1(r)=\{x:x \gg r\}$,   $K^{2}(r)=\{x:x_1<r_1\text{ and }x_2>r_2\}$, $K^{3}(r)=\{x:x_1>r_1\text{ and }x_2<r_2\}$,  and $K^{4}(r)=\{x:x \ll r\}$  formally defines  the four categories described above.

In the absence of losses, the DM  values each alternative with an additive utility function, $ u(x_1)-u(r_1)+ v(x_2)-v(r_2)$, which attaches equal weight to each attribute. If she experiences a loss in attribute $i$, then she inflates the weight attached to that attribute by  $\lambda_i $.
Then, the utility function is:
\begin{equation*}
V_{TK}(x|r)=\left\{
\begin{array}{ll}
u_1(x_1)-u_1(r_1)  + u_2(x_2) -u_2(r_2)  & \text{ if  } x \in K^1(r)  \\
\lambda_1 (u_1(x_1)-u_1(r_1))  +u_2(x_2) -u_2(r_2)  & \text{ if  }  x \in K^2(r)\\
u_1(x_1)-u_1(r_1)  +  \lambda_2 (u_2(x_2) -u_2(r_2))  & \text{ if }  x \in K^3(r)\\
\lambda_1 (u_1(x_1)-u_1(r_1))  +  \lambda_2 (u_2(x_2) -u_2(r_2))  & \text{ if  } x \in K^4(r) \\
\end{array}%
\right. 
\end{equation*}	
where $\lambda_1,\lambda_2>0$ ($>1$ if loss averse) and each $u_i$ strictly increases. TK is a special case of Affine CTM with four categories defined by a gain-loss category function.
\subsubsection*{Status Quo Bias Model (MO)}  Particularly for difficult decisions, rejecting the status quo for another alternative causes psychological discomfort, unless that alternative is unambiguously superior to it  (see \cite{fleming2010overcoming}). 
People categorize alternatives according to whether they are obvious improvements, and tend to stick to a suboptimal status quo, particularly when the trade-off is unfamiliar and unclear. 
\cite{MO05} introduce this  concept to economics by modeling individuals who incur
an additional utility cost when they abandon the status quo for something not obviously better. 
 
\cite{MO05} derive a closed set  $Q(r)$ that denotes the alternatives that are unambiguously superior to the default option $r$ which include but are not limited to those that exceed $r$ in all attributes  (Figure \ref{fig:models}).  This set formally maps to a category function $\mathcal{K}^{MO}=(K^1,K^2)$ where 
 $K^1(r)=\{x | \ x \in int(Q(r))\}$  and $K^2(r)=\{x | \ x \notin Q(r)\}$. 
The former contains all those bundles obviously better than the status quo, and the latter contains those that are not. Here, we consider a special case of their model. If an alternative is not obviously better than the status quo, then the DM pays a cost $c(r)>0$ to move away from the status quo, which may depend on the reference point. For any $x \neq r$, we have:
\begin{equation*}
\label{eq:MO rep}
V_{MO}(x|r)=\left\{ 
\begin{array}{lc}
u_1(x_1) + u_2(x_2)  & \text{ if } x \in K^1(r)  \\
u_1(x_1) + u_2(x_2) - c(r)  & \text{ if  }  x \in K^2(r)\end{array}%
\right. .
\end{equation*}
This is an example of an Affine CTM for general $c$, and a Strong CTM when $c(r)$ is constant.
\subsubsection*{Salient Thinking Model (BGS)}
The context in which a decision takes place causes some features of an alternative to stand out, making them more salient than others. 
When a portion of the alternative is more salient,   
psychologists have found that ``the information contained
in that portion will receive disproportionate weighing in
subsequent judgments'' \citep{taylor1982stalking}.
That is, people unconsciously categorize goods according to which of their features is most salient.
\cite{bordalo2013salience} propose a behavioral model based on salience and show that it has a number of important consequences. 

In the model,  a salience function $\sigma:\mathbb{R}_{++}\times\mathbb{R}_{++}\rightarrow \mathbb{R}_+ $ determines the salience of a given attribute of an alternative.%
\footnote{We describe the properties of $\sigma$ more fully in Section \ref{sec: BGS}.}
Formally, the salience category function $\mathcal{K}^{BGS}=(K^1,K^2)$ when  $
K^1(r)=\{x:\sigma(x_1,r_1)>\sigma(x_2,r_2)\}$ and $K^2(r)=\{x:\sigma(x_1,r_1)<\sigma(x_2,r_2)\}$. This function
indicates which alternatives have each salient attribute.
In words, the DM categorizes objects according to the attribute that differs the most from the reference point according to the salience function, and $K^i(r)$ is the set of those for which attribute $i$ stands out the most.
That is, given a reference $(r_1,r_2)$, attribute 1 is salient for good $x$ if  $\sigma( x_1, r_1)  > \sigma( x_2, r_2) $, and attribute 2 is salient for good $x$ if  $\sigma( x_1, r_1)  <  \sigma( x_2, r_2) $. BGS propose  the salience function  $\sigma(x_k,r_k)=\frac{|x_k-r_k|}{x_k+r_k}$, and we illustrate the indifference curves based on this function in Figure \ref{fig:models}.

An attribute receives more weight when it is salient than when it is not. The family $\{\succsim_r\}_{r\in X}$ has a \emph{BGS $(\sigma;w_1,w_2,u_1,u_2)$ representation} if each $\succsim_r$ is represented by
	\begin{equation*}
	\label{eq:BGS rep}
	V_{BGS}(x|r)=\left\{ 
	\begin{array}{cc}
	w^1_1 u_1(x_1) +  w^1_2 u_2( x_2)  & \text{ if } x \in K^1(r)\\
	w^2_1 u_1( x_1) +  w^2_2 u_2( x_2)  & \text{ if } x \in K^2(r)\\
	\end{array}%
	\right. 
	\end{equation*} 
	for a salience function $\sigma$ with strictly positive weights with $\frac{w^1_1}{w^1_2}>\frac{w^2_1}{w^2_2}$, and each $u_i$ strictly increases.
	 Because $\frac{w^1_1}{w^1_2}>\frac{w^2_1}{w^2_2}$, the DM is less willing to trade-off less of attribute 1 for more of attribute 2 when attribute 1 is salient than when it is not. Consequently, alternatives relatively strong in the first dimension improve when categorized as $1$-salient, but those relatively strong in the second are hurt.
	 
\subsubsection*{Prototype Theory (PT)} 
A key role of categorization is to simplify the representation of a complex environment. People evaluate objects categorized in the same way according to similar criteria, and one way in which psychologists explain category formation is through prototype theory \citep{posner1970retention}. It argues that people categorize a stimulus according to how similar it is to a prototype that is the ``most typical'' member of the category. 
As \citet[p. 36]{Rosch1978}  argues,  ``Categories can be viewed in terms of their clear cases  if the perceiver places emphasis on the correlational structure of perceived attributes.''
We propose a model of choice based on these ideas. The DM compares each alternative  to each prototype and categorizes it accordingly. Then, she evaluates it according to how it differs from the prototype.

In the model, there are $m$ prototypes, $p^1,\dots,p^m$, and the DM categorizes each alternative according to  how close it is to a prototype.
Then, category $K^i(r)$ is the set of alternatives most similar to $p^i$ and, as suggested by \cite{tversky1978studies}, similarity may depend on the reference point.
Formally, there is a family of metrics indexed by $r$ so that $d_r(x,y)$ indicates how far away the DM perceives $x$ to be from $y$ given  $r$; each is continuous with respect to the usual metric on $X$.%
\footnote{This metric could be replaced by similarity function, as proposed by\cite{tversky1977features}, without changing any of the key insights.}
The category function  $\mathcal{K}^{P}=(K^1, \dots, K^m)$ where 
$K^i(r)=\left\{x: i = \arg \min_j d_r(p^j,x) \right\}$.
The DM evaluates alternatives in category $i$ according to:
$$V_{PT}(x|r)=U(p^i)+\lambda^i_1(x_1-p^i_1)+\lambda^i_2(x_2-p^i_2)\text{ if } x \in K^i(r)$$
where $U(\cdot)$ is a hedonic utility function and $\lambda^i_j>0$.
A particularly interesting specification is when $\lambda^i_j= \frac{\partial}{\partial p^i_j}U(p^i)$.
Then, the DM approximates the utility of $x$ according to a first-order Taylor expansion around the prototype most similar to it (Figure \ref{fig:models}).\footnote{In Figure \ref{fig:models}, we use $d_r(p^j,x)=r_1|x_1-p^j_1|+r_2|x_2-p^j_2|$ to illustrate this model.} 
This is an example of a Strong CTM.
\subsection{Time Preference}

We can apply our framework to  time preference where each alternative corresponds to an amount of consumption at a given point in time. 
People treat future outcomes differently than immediate outcomes  \citep{frederick2002time}. 
\citep{McClure2004} documents physiological reasons for this distinction: decisions involving immediate trade-offs are associated with the limbic system, but the
prefrontal and parietal regions are active in decisions involving future trade-offs. Consequently, people  categorize rewards as being short-term or long-term, and many suffer from present-bias; that is, they are less patient for those in the former category than those in the latter.  Economists have employed quasi-hyperbolic discounting \citep{phelpspollak1968} to capture this behavior. In Appendix \ref{appendix:otherCTM}, we formally illustrate that this model is a special case of CTM when what is ``present'' depends on a reference outcome.
\subsection{Social preferences}
Our final application is to other-regarding preferences where alternatives represent allocations of consumption to each of $n$ individuals. Two leading models of social preferences, the inequity aversion model of \cite{fehrschmidt99} and the distributional preferences model of \cite{charness2002understanding}, fit under the umbrella of CTM.
These models implicitly depend on the outcomes that the DM expects for herself and others, which we model as a social reference point.
While most take the equitable outcome as the social reference point, Fehr and Schmidt note that ``[t]he determination of the relevant reference group and the relevant reference outcome for a given class of individuals is ultimately an empirical question'' (page 821), and that it may depend on, among other things, the social context.
In the inequity aversion model, people categorize social allocations according to whether their inequities are advantageous or disadvantageous; the former causes them to experience envy of an individual's allocation, while the latter experiences guilt.
The distributional preferences model focuses on people's trade-off between their own material payoff and overall social welfare, where welfare includes both a utilitarian component and one that focuses on the utility of the worst-treated person.
Consequently, individuals categorize social allocations according to the identity of the worst-treated individual.
In the appendix, we discuss these models with a general reference point and show that they fall under the umbrella of CTM.

\section{Behavioral Foundation for CTM}
\label{sec: axiomatization}
In this section, we provide a set of behavioral postulates characterizing increasing CTM. These postulates represents the key features of the model. We show that they hold if and only if the data is representable by increasing CTM, rendering the model behaviorally testable. In subsequent subsections, we explore the various strengthenings of the model and provide axiomatizations of these as well.

Our axioms apply to the family $\{\succsim_r\}_{r \in X}$, an  observable object, and are stated in terms of $\mathcal{K}$,  a component of the model.
They inform us whether that family is CTM with category function $\mathcal{K}$. In other words, the axioms answer the question ``Given $\mathcal{K}$, are there category utilities for which $\{\succsim_r\}_{r \in X}$ is CTM under $\mathcal{K}$?'' Since the categories themselves convey much of the psychology captured by CTM, they tie the behavior of the DM, as reflected by $\{\succsim_r\}_{r \in X}$, to the phenomenon to be captured. Of course, this leaves open the possibility that the family  $\{\succsim_r\}_{r \in X}$ violates the axioms for one category function $ \mathcal{K}$ but satisfies them for  $\mathcal{K}'$. Since $\mathcal{K}$ and $\mathcal{K}'$ would presumably be derived under different theories of behavior, the axioms inform which, if either, of the two describes the DM's choices.  

Another important question is ``Are there category utilities and a $\mathcal{K}$ for which $\{\succsim_r\}_{r \in X}$ is CTM?''  In Section \ref{sec:revealing}, we discuss how to answer this question. We show that under mild conditions, one can first infer the DM's endogenous category function $\hat{\mathcal{K}}$. One could then test the axioms on $\{\succsim_r\}_{r \in X}$ for $\hat{\mathcal{K}}$. Proposition \ref{prop: revealing regions} and Corollary \ref{cor: endog BGS} illustrate how to do so for BGS.

Define the revealed ranking within category $k$ as the  binary relation $\succsim^k$ for which $x \succsim^k y$ if and only if there exists $r$ such that $x,y\in K^k(r)$ and $x \succsim_r y$. The sub-relations $\succ^k$ and $\sim^k$ are defined in the usual way.  The ranking $\succsim^k $ captures preference within category $k$. The following axiom states that the within-category revealed preference  has no cycles.

\begin{ax}[Weak Reference Irrelevance]
	\label{ax: first}
	The relation $\succsim^k$ is acyclic. That is, if $x^1 \succsim^k x^2 \succsim^k \dots \succsim^k x^m$, then  $x^m \not \succ^k x^1$. 
\end{ax}

Weak Reference Irrelevance ensures that the DM reacts consistently to alternatives when they are categorized the same way.
That is, the categories reflect the DM's psychological treatment of the alternative.
Although she may have choice cycles, these cycles occur only when the context changes how the DM categorizes alternatives.
Since $\succsim^k$ is acyclic, we can take its transitive closure to derive full comparisons.
Let $\succsim^{k*}$ be is transitive closure, with $\succ^{k*}$ and $\sim^{k*}$ the asymmetric and symmetric parts.

Within a category, preference has an additive structure. The next axiom implies that each $\succsim_r$ satisfies Cancellation when restricted to a given category.
\begin{ax}[Category Cancellation]
	\label{ax: CDC}
		For all $x_1, y_1, z_1 , x_2,y_2,z_2\in \mathbb{R}_{++}$, $r \in X$, and category $j$ so that $(x_1,z_2), (z_1,y_2),(z_1,x_2), (y_1,z_2),(x_1,x_2) ,  (y_1, y_2) \in K^j(r)$:\\		
		 If $(x_1,z_2) \succsim_r (z_1,y_2)$ and $(z_1,x_2) \succsim_r (y_1,z_2)$, then $(x_1,x_2) \succsim_r  (y_1, y_2)$.
\end{ax}

Category Cancellation adapts the well-known Cancellation axiom to our setting, differing in its requirement that the alternatives belong to the same category. Without the qualifiers on how alternatives are categorized, the axiom is a well-known necessary condition for an additive representation that appears in \cite{krantz1971} and \cite{TK91}, among others.    If $X$ has strictly more than two dimensions, then we can replace it with the analog of P2  \citep{Savage1954}; see  \cite{debreu1959topological}.%
\footnote{Formally, for any $x,y,x',y' \in K^k(r)$  and subset of indexes $E$, if $x_i=x'_i$ and $y_i=y'_i$ for $i \in E$,  $x_i=y_i$ and $x'_i=y'_i$ for all $i \notin E$, and $x\succsim_r y$, then  $x' \succsim_r y'$. This is implied by Category Monotonicity when $n=2$, so a stronger condition is necessary.} 

The next axiom requires that Monotonicity holds between objects categorized the same way.
\begin{ax}[Category Monotonicity (CM)]
	\label{ax: CM}
	For any $x,y,r \in X$: if $x \geq y$ and $x \neq y$, then $y \not \succsim^{k*} x$ for any category $k$; in particular, if $x,y \in K^k(r)$, then $x \succ_r y$.
\end{ax}
 Since both attributes are ``goods'' as opposed to ``bads,'' Monotonicity means that if a product $x$ contains more of some or all attributes, but no less of any, than another product $y$, then $x$ is preferred to $y$. The postulate requires that choice respects Monotonicity for alternatives  within the same category. However, it does  not require that this comparison holds when the goods belong to different categories, and we shall see later that salience can distort comparisons enough to cause Monotonicity violations.
 
 Finally, the family of preference relations is suitably continuous.
\begin{ax}[Category Continuity]
	\label{ax:last_CTM}
	For any $r \in X$, $x \in \bigcup_i K^i(r)$, and category $j$, the sets $UC_j(x)=\{y \in K^j(r):y \succ_r x\}$ and $LC_j(x)=\{y\in K^j(r):x \succ_r y\}$ are open. Moreover, the set  $$\left\{x \in \bigcup_i K^i(r):UC_j(x)\bigcup LC_j(x)=K^j(r)\text{ and }UC_j(x)\neq K^j(r) \text{ and }LC_j(x)\neq K^j(r) \right\} $$ has an empty interior.
\end{ax}
Category  continuity adapts the usual continuity condition to apply only within a category. It says that when $y$ is preferred to $x$ in a given context and $y'$ is close enough to $y$, then $y'$ is also preferred to $x$, provided that $y'$ belongs to the same category as $y$. The final condition requires that if an alternative $x$ is neither better than everything within category $j$ nor worse than everything within category $j$, then there exists something in category $j$ that is as good as $x$, or as good as something arbitrarily close to $x$.
For such an $x$, the category must intersect almost all indifference curves close to $x$'s since each category is almost connected.

Finally, we make a structural assumption.
\begin{ass*}[Structure]
	The category function $\mathcal{K} $ is such that for any category $k$, the following sets are connected: $E^k=\bigcup_{r\in X} K^k(r)$, $\{x\in E^k: x_i=s\}$ for all dimensions $i$ and scalars $s$, and $\{y \in E^k: x \sim^{k*} y \}$ for all $x\in E^k$. 
\end{ass*}

The Structure Assumption is satisfied for all the models we discussed in the previous section. Indeed, $E^k=\mathbb{R}^n_{++}$ for every category $k$ in each of these models, except prototype theory.%
\footnote{For instance, $p^k \notin E^j$ for every $j \neq k$. We thank a referee for pointing this out.}
These conditions establish that the objects categorized in the same way have enough topological structure so that ``local'' properties can be extended to global ones.
\cite{chateauneuf1993local}  show that the structure assumption, applied to a single preference relation and domain, is needed to guarantee that a local additive representation implies a global one.

\begin{thm}
	\label{thm: weak CTM}
	Assume the Structure Assumption holds. The family $\{\succsim_r\}_{r\in X}$ satisfies Weak Reference Irrelevance, Category Cancellation, Category Monotonicity, and Category Continuity  for $\mathcal{K} $  if and only if it conforms to increasing CTM under $\mathcal{K} $.
 
\end{thm}

Increasing CTM captures the behavior implied by the axioms, so we call Axioms \ref{ax: first}-\ref{ax:last_CTM} the CTM axioms.
Taken together, they establish that the DM acts rationally when restricting attention to alternatives categorized in the same way for a given reference point.
That is, CTM captures a DM who differs from the neoclassical model only when alternatives are categorized differently.
The theorem reveals that a number of other reference dependent models have been studied by the literature fall outside the scope of our analysis.
For instance, \cite{bhatia2013attention}, \cite{munro2003theory}, the non-constant loss averse version of \cite{TK91}, and the continuous version of the salient thinking model (see online appendix of \cite{bordalo2013salience} and the related  \cite{BGS_MAC}) all violate weak reference irrelevance for any specification of the category function.
We provide the details in Appendix \ref{sec: behavioral models that are not RPM}.

We provide a brief outline of how the proof works, and all omitted proofs can be found in the appendix. The axioms are sufficient for a ``local'' additive representation of $\succsim_r$ (and thus $\succsim^k$) on an open ball around each alternative within category $k$.
The Structure Assumption allows us to apply Theorem 2.2 of \cite{chateauneuf1993local}  to aggregate the local additive representation of $\succsim^k$ into a global one.  To do so, we must establish that the global preference is complete, transitive, monotone, and continuous.
We establish these properties for preference within each category by  showing that the transitive closure of each $\succsim^k$ is complete and suitably continuous.
The remainder of the proof shows that Categorical Continuity allows us to stitch the different within-category representations together into an overall utility function.

\subsection{Reweighting}

In all of the models discussed in Section \ref{sec: consumer choice}, the DM evaluates the difference between alternatives categorized in the same way similarly.
That is, regardless of the category, the DM agrees on how much better a value of $x$ versus $y$ is in dimension $i$.
Categorization affects only how much weight she puts on each dimension.
This is captured by the following axiom.

\begin{ax}[Reference Interlocking]
	\label{ax: ref_inter}
	For any $a,b,a',b',x',y',x,y\in  X$ and categories $k,j$ with
	$x_{-i}=a_{-i}$, $y_{-i}=b_{-i}$,  $x'_{-i}=a'_{-i}$, $y'_{-i}=b'_{-i}$, $x_i=x'_i$, $y_i=y'_i$, $a_i=a'_i$, $b_i=b'_i$:\\
	if $x \sim^k y$, $a \succsim^k b$, and $x' \sim^j y'$, then  it does not hold that $b' \succ^j a'$.
\end{ax}

The term ``Reference Interlocking'' comes from \cite{TK91}. If  each $\succsim^k$ is complete, then their statement of it is equivalent given the other axioms. 
Roughly, the DM agrees on the difference in utilities along a given dimension regardless of how an alternative is categorized.
To interpret, observe that the first pair of comparisons reveals that the difference between $a_i$ and $b_i$ exceeds that between $x_i$ and $y_i$ when the alternatives belong to category $k$.
For alternatives categorized in $j$, the DM should not reveal the opposite ranking.
We defer to the above paper for a detailed discussion.

\begin{thm}
	\label{thm: RI}
	Suppose that $\{\succsim_r\}_{r\in X}$ conforms to increasing  CTM under $\mathcal{K}$ and each $E^k$ is connected.
	For each dimension $i$, there exist a utility index $u_i$  and a weight $w^k_i>0$ for each category $k$ so that each category utility $U^k$ is cardinally equivalent to one that maps each $x \in E^k$ to $\sum_i w^k_i u_i(x_i) $ if and only if Reference Interlocking holds.

\end{thm}

All of the models in Section \ref{sec: consumer choice} satisfy the axiom, and are thus special cases of increasing CTM satisfying Reference Interlocking.
For instance, differences in the salient dimension of BGS receive higher weight, but the relative size of two given differences in the same dimension is the same regardless of whether both are salient or both are not.
The axiom implies that the utility index within each category must be the same, up to an increasing, affine transformation.

\subsection{Behavioral Foundation for  Affine CTM}
In this section, we explore when an Affine CTM exists. That is, we provide conditions under which $U^k(\cdot|r)$ a positive affine transformation of $U^k(\cdot|r')$ for any $r,r'$. All of the models from Section \ref{sec: consumer choice} fall into this class.%
\footnote{For MO, this is true only when $c(r)<\infty$.} 

Unsurprisingly, the key restriction relative to CTM is that tradeoffs across categories are affine. As is usual, this is captured by a form of linearity, or the ``Independence Axiom.''  We require it to hold only when alternatives combined belong to the same category, and adjust for the curvature of the utility index.

To state the key axiom, we define an operation $\oplus^k$ along similar lines as \cite{ghirardato2003subjective}.  For $x,y\in \mathbb{R}$ and a category $k$, $\frac12 x\oplus^k_i \frac12 y=z$ when there exists $a,b$ such that $(x_i,a_{-i})\sim^{k*} (z_i,b_{-i})$ and $(z_i,a_{-i})\sim^{k*} (y_i,b_{-i})$.
If $\succsim^k$ has an additive representation, then $\frac12 U^k_i(x)+ \frac12 U^k_i(y)= U^k_i(z)$. 
Define $\oplus^k$ similarly for alternatives: $\frac12 x \oplus^k \frac12 y=z$ if and only if $z_i =\frac12 x_i\oplus^k_i \frac12 y_i$ for each dimension $i$.
Finally, define $\alpha x \oplus^k (1-\alpha)y$ by taking limits.%
\footnote{In general, $\alpha x \oplus^k (1-\alpha)y$ need not exist. However, it does exist ``locally,'' which is all we require in the proof. That is, if $x \in K^k(r)$, then there exists an open set $O$ with $x\in O$ on which $\alpha y \oplus^k (1-\alpha)z$ exists for every $\alpha \in [0,1]$ and $y,z \in O$.}
We note that if $U^k_i$ is linear, then $\alpha x \oplus^k_i (1-\alpha) y=\alpha x+(1-\alpha) y$.

\begin{ax}[Affine Across Categories (AAC)]
	\label{ax: AAC}
	For any $r \in X$, $x,x',\alpha x \oplus^j (1-\alpha) x' \in K^j(r)$, and $y,y',\alpha y \oplus^k (1-\alpha) y' \in K^k(r)$: if $x \succsim_r y$ and $x' \succsim_{r} y'$, then $\alpha x \oplus^j (1-\alpha) x' \succsim_r \alpha y \oplus^k (1-\alpha) y'$. 
\end{ax}

This axiom is a natural adaptation of the linearity axiom, a close relative of the independence axiom.
If we strengthened Affine Across Categories to be stated using the traditional linearity condition, then we would obtain a representation where each $U^k(\cdot|r)$ is itself an affine function.
Otherwise, it requires that the $\oplus^k$ operation preserves indifference.

The second axiom deals with a technical issue.  
\begin{ax}[Unbounded]
	\label{ax: unbounded}
	For any $r\in X$: if $K^k(r)$ contains a sequence $x_n$ so that $U^k(x_n) \rightarrow \infty$ ($-\infty$), then for any $x \in X$ there exists $x^*\in K^k(r)$ so that $x^* \succ_r x$ ($x \succ_r x^*$).
\end{ax}

We note that $U^k$ is unique up to a positive affine transformation. Hence whenever the utility of some sequence goes to infinity for some representation of $\succsim^k$, it must also converge to infinity for any other representation as well.
While the axiom can be stated in terms of primitives, 
we instead state it in terms of the $U^k$.%
\footnote{The statement in terms of primitives involves standard sequences and does not reveal key aspects of behavior, so we instead present the simpler and easier to interpret one above.
	In special cases, this is easy to do. For instance, if $U^k$ is linear, then the axiom simply states that if $K^k(r)$ is an unbounded set, then the conclusion of the above axiom holds.}
It ensures that a category containing alternatives whose utility goes to positive (negative) infinity must contain an alternative better (worse) than any other given alternative.
If it failed, then no affine transformation of the category utility would represent the preference.


\begin{thm}
	\label{result: Affine CTM}
	Assume the Structure Assumption holds.
	Then, 
	$\{\succsim_r\}_{r\in X}$ satisfies the CTM axioms,  Affine Across Categories, and Unbounded for $\mathcal{K}$ if and only if it conforms to Affine Increasing CTM  under $\mathcal{K}$.
\end{thm}

	All the models discussed in Section \ref{sec: Model} fall into the class of Affine CTM, so the result reveals the behavior all have in common. 	Relative to CTM, Affine Across Categories imposes stronger requirements on how the DM relates alternatives in different categories. Not only does the DM evaluate utility within a category using an additive function, but the additive structure persists across categories.	Moreover, this aids with interpreting utility differences. If every pair of categories contains alternatives indifferent to one another, the entire representation is unique up to a common positive affine transformation. We call the combination of Axioms \ref{ax: first}-\ref{ax:last_CTM} and \ref{ax: AAC}-\ref{ax: unbounded} the Affine CTM axioms.

\subsection{Behavioral Foundation for Strong CTM}

For a Strong CTM, changing the reference point does not reverse the ranking of two products unless it also changes their categorization.
The following axiom imposes this.

\begin{ax}[Reference Irrelevance] For any $x,y,r,r'\in X$:\\ 
	\label{ax:SC}
	if  $x\in K^k(r) \bigcap K^k(r')$ and $y\in K^l(r) \bigcap K^l(r')$, then  $x \succsim_{r} y$ if and only if  $x \succsim_{r'} y$.
\end{ax}

For the general CTM, the reference point influences choice trough two channels: the category to which it belongs and its valuation. The axiom eliminates the latter.
When comparing two alternatives across different reference points, the DM's relative ranking does not change when neither's category changes. This property greatly limits the effect of the reference point.  In fact, a sufficiently small change in the reference never leads to a preference reversal. 
\begin{thm}
	\label{result: strong CTM}
	
	Assume the Structure Assumption holds and for any categories $i,j$ and any $ r \in X$, there exists $x \in K^i(r)$ and $y  \in K^j(r)$ with $x \sim_{r} y$.
	Then, $\{\succsim_r\}_{r\in X}$ satisfies the Affine CTM axioms and Reference Irrelevance for $\mathcal{K}$  if and only if conforms to Strong, Increasing CTM under $\mathcal{K}$.
\end{thm}

Since BGS, MO, and PT are Strong CTM, Theorem \ref{result: strong CTM} characterizes the behavior they have in common. While the reference  plays a role in categorization, it plays no role in choice after categorization is taken into account.
TK, which belongs to Affine CTM but not Strong CTM, must therefore violate reference irrelevance.

\subsection{Comparing Models of Riskless Choice}
	TK, BGS, MO, PT,  and the neoclassical model all conform to Affine CTM, so Theorems \ref{thm: weak CTM}  and \ref{result: Affine CTM} describe the behavior that they have in common. 
	However, the analysis so far, as well as the functional forms of the models, leaves open the question of what behavior distinguishes them.
	Of course, they differ in how alternatives are categorized, but the models also reflect distinct behavior within and across categories. 
	
	In addition to Reference Irrelevance, they are distinguished by whether they satisfy two classic axioms: Monotonicity and Cancellation, the unrestricted versions of Category Monotonicity and Category Cancellation.%
	\footnote{The formal statements are obtained by dropping the requirement in those two axioms that the alternatives belong to the same category.}
	The first requires that a dominant bundle is chosen, and the latter that an additive structure obtains.
	The representation theorem of \cite{TK91}  imposes those two axioms in addition to continuity. In Appendix \ref{pf: thm:TKinRPM}, we show that an Affine CTM with a Gain-Loss category function satisfies the two classic axioms and continuity if and only if it has a TK representation.
	We provide a detailed examination of the BGS model in Section \ref{sec: BGS}.
	
	Table \ref{table:comparison}  compares the four models in terms of Reference Irrelevance, Monotonicity and Cancellation, when BGS, TK, MO, and PT do not coincide with the neoclassical model. 
	Only the neoclassical model satisfies all conditions; none of the other four do.   On the one hand, BGS and PT satisfy Reference Irrelevance but violate Monotonicity and Cancellation. On the  other, TK maintains Monotonicity and Cancellation but violates Reference Irrelevance.  Finally, MO satisfies all but Cancellation.\footnote{Propositions \ref{thm:BGS preference} and \ref{thm:TKinRPM} give the \cmark's of the table for BGS and TK. It is routine to verify that MO satisfies Monotonicity and Reference Irrelevance and the PT satisfies RI. We provide examples showing the other properties are violated in Appendix \ref{sec: examples table}.}
	
	\begin{table}[h!]
		\begin{center}
			\begin{tabular}{lc|ccccc}
				\hline \hline
				~                  & ~     & Neoclassical & \ BGS  \ & \  \ TK \ \ & \ \ MO \ \ &\ \  PT\ \  \\ \hline 
				CTM       & ~     & \cmark              &  \cmark   & \cmark  & \cmark & \cmark \\ \hline
				Monotonicity       & ~     & \cmark              &  \xmark   & \cmark  & \cmark & \xmark \\ \hline
				Reference Irrelevance & ~     & \cmark               & \cmark  &  \xmark  & \cmark%
				\footnotemark & \cmark \\ \hline
				Cancellation 	   & ~	   & \cmark 	          &  \xmark	& \cmark  &  \xmark & \xmark  \\   \hline \hline
			\end{tabular}
		\end{center}
		\caption{Comparisons of Models}\label{table:comparison}
	\end{table}
	
	\footnotetext{Whenever $c(r)=c(r')$ for every $r,r' \in X$.}
	We provide a plausible example violating the Cancellation axiom, and hence behavior inconsistent with TK. Then, we illustrate BGS can accommodate this example even without requiring a shift in the reference point. While the example is one simple test to distinguish BGS from TK, it is also powerful as it works for a fixed reference point.

	\begin{example} 
		\label{ex: BGS cancellation}
		\textit{Consider a consumer who visits the same wine bar regularly. The bartender occasionally offers promotions. The customer prefers to pay $\$8$ for a glass of French Syrah rather than $\$2$ for a glass of Australian Shiraz. At the same time, she prefers to pay $\$2$ for a bottle of water rather than $\$10$ for the glass of French Syrah. However, without any promotion in the store, she prefers paying $\$10$ for Australian Shiraz to paying $\$8$ for water.}
	\end{example}

	The behavior in this example is both intuitively and formally consistent with the salient thinking model of BGS.\footnote{Implicitly, the example reveals that the quality of French Syrah is higher than Australian Shiraz which is in turn higher than water. The numerical value of quality assigned to each beverage is irrelevant to the violation of Cancellation. For examples of qualities so that choice can be represented by the BGS model, one can calculate that $(-8,q_{fs}) \succ_r (-2, q_{as})$, $(-2,q_{w}) \succ_r (-10, q_{fs})$ and $(-10,q_{as}) \succ_r (-8,q_w)$ for $q_{fs}=8$, $q_{as}=6.9$,  $q_w=5.1$, and the reference point $r=(\frac12(-10+-8),\frac12 (q_w +q_{as}))$ when $w=0.6$.}
	Without any promotion, the consumer expects to pay a high price for a relatively low quality selection. 
	When choosing between Syrah or Shiraz, the consumer focuses on the French wine's sublime quality, and she is willing to pay  at least $\$ 6$ more for it. When choosing between water and Syrah, the low price of water stands out and she reveals that the gap between wine and water is less than $\$8$. However, when there is no promotion, she focuses again on the quality, and she is willing to pay an additional $\$ 2$ for even her less-preferred Australian Shiraz over water. Notice that this explanation does not require that the reference points are different. Since the consumer visits this bar regularly, intuitively, her reference point should be fixed and stable.

\subsection{Revealing categories}\label{sec:revealing}

Up to now, we have taken the category function as known.
This subsection explores the extent to which one can infer categories directly from choices. We  first show this can be done when the categorization of the object alters the trade-offs between attributes, so local behavior directly reveals how an object is categorized. 
Finally, we outline an alternative approach applicable in the presence of discontinuities across categories, even when trade-offs are unaffected by categorization.

Our identification of the categories is based on local indifference sets (LIS).
For a CTM with categories $k$ and $l$, we  write $LIS^k(x)=LIS^l(x)$ if there exists a neighborhood $O$ of $x$ so that 
$$U^k(y)=U^k(x) \iff U^l(y)=U^l(x) \text{ for all }y\in O;$$
otherwise, $LIS^k(x)\neq LIS^l(x)$.
If $LIS^k(x)=LIS^l(x)$, then any alternative indifferent to $x$ when $x$ is categorized as $k$ is also indifferent to $x$ when it is categorized as $l$, provided that it is not too far away from $x$. In neoclassical consumer theory  with sufficiently differentiable utility, this is equivalent to the marginal rate of substitution at the bundle $x$  being equal across categories. 
Put another way, the trade-off between the attributes does not depend on how the alternative is categorized. 
If $LIS^k(x)\neq LIS^l(x)$, then categorization affects trade-offs between attributes. We require the latter, i.e., a different pattern of substitution within each category.

\begin{proposition}
\label{prop: general category ID} Let $\{\succsim_r\}_{r\in X}$ be  a CTM. For any category $k$ such that $LIS^k(x)\neq LIS^l(x)$ for every $x \in X$ and category $l \neq k$, category $k$ is uniquely identified.
\end{proposition}

The result shows that categories are uniquely identified whenever the DM makes different trade-offs at every alternative for different categories. Moreover, if the assumption of Proposition \ref{prop: general category ID} holds for all categories, we can reveal all the categories. Proposition \ref{prop: revealing regions} shows the result is always applicable to the salient thinking model. 
Moreover, it applies to prototype theory whenever $\lambda^k$ is not a rescaling of $\lambda^l$ for any $k \neq l$, and to TK for the gain-loss and loss-gain regions whenever $\lambda_1, \lambda_2 \neq 1$.

For an intuition, suppose that the category utilities are affine (i.e., $U^k(x)= \sum_i u^k_i x_i +\beta^k$), so indifference curves are (piecewise) straight lines. Then, $LIS^k(x) \neq LIS^l(x)$ whenever their slope within category $k$ differs from the slope within $l$ are different.
Examining the DM's choices between alternatives close to $x$ allows us to identify the slope of the indifference curve at that point, and hence whether $x$ belongs to category $k$ or to $l$.
%
%

While almost all CTM satisfy the hypothesis of Proposition \ref{prop: general category ID}, some do not, such as MO (see Figure \ref{fig:models}). Although we cannot apply the above result to identify their categories, one can also identify the categories by utilizing discontinuities at the border.
To illustrate, in the MO model the utility of an alternative sharply drops when it is no longer unambiguously better than the status quo. This leads to a discontinuity in the indifference curve, and these discontinuities trace out the boundary between the categories.
We discuss how this argument generalizes in Appendix \ref{sec: prop: revealing regions}, and also provide an example where distinguishing categories from choice is impossible.

%

\section{BGS Model and Categories}
\label{sec: BGS}

The salient-thinking model accounts for a number of empirical anomalies for the neoclassical model with a single, intuitive mechanism. Despite its popularity, it can be difficult to understand all of the model's implications for choice: its new components, in particular the salience that determines which attribute stands out for a given reference point, are unobservable.
This section uses CTM to provide a characterization of the choice behavior implied by BGS.    
We begin by studying  the properties of the categories generated by the salience function.

We say $\sigma$ is a \emph{salience function} if it satisfies four basic properties: i) it increases in contrast: for $\epsilon>0$ and $a>b$, $\sigma(a+\epsilon,b)>\sigma(a,b)$   and $\sigma(a,b-\epsilon)>\sigma(a,b)$;   ii)  it  is continuous in both arguments; iii)  it is symmetric: $ \sigma(a,b)= \sigma(b,a)$; and iv) it is grounded: $\sigma(r,r)=\sigma(r',r')$ for all $r,r'\in X$.  
Two other properties are often imposed: $\sigma$ is   Homogeneous of Degree Zero (HOD) if for all $\alpha > 0$,   $\sigma(\alpha a, \alpha b) =  \sigma( a, b)$, and $\sigma$ has   diminishing sensitivity if for all $\epsilon > 0$ and $a,b>0$,   $\sigma( a+\epsilon,  b+\epsilon) \leq  \sigma( a, b)$.\footnote{Requiring a strict inequality is problematic. If, as usually assumed, $\sigma$ is HOD, then $\sigma(r,r)=\sigma(\alpha r, \alpha r)=\sigma(r+\epsilon,r+\epsilon)$ for $\alpha>1$ and $\epsilon=(\alpha-1)r$, a  contradiction.}  
The first three  properties of the salience function are explicitly stated by \cite{bordalo2013salience},  and the fourth is  satisfied by all of the specifications in the literature. It is a necessary condition for an attribute to be salient for a good only if it differs from the reference good in it.

Consider the following properties of categories. 

\begin{description} 
	\item[S0] (Basic) For any $r\in X$: $K^1(r) \bigcap K^2(r)=\emptyset$,   $K^1(r) \bigcup K^2(r)$ is dense in $X$,  $K^1$ and $K^2$ are continuous in $r$, and $K^1(r)$ and $K^2(r)$ are regular open sets.
	
	\item[S1] (Moderation) For any $\lambda \in [0,1]$ and $r\in X$: if $x \in K^{k}(r)$, $y_{k}=x_{k}$, and $y_{-k} =\lambda x_{-k} +(1-\lambda) r_{-k}$, then $y \in K^{k}(r)$.
	
	\item[S2] (Symmetry) If $(a,b) \in K^k(c,d)$, then $(c,d) \in K^k(a,b)$ and $(b,a) \in K^{-k}(d,c)$.
	
	\item[S3] (Transitivity) If $(a_1,a_2) \notin K^2(r_1,r_2)$ and $(a_2,a_3) \notin K^2(r_2,r_3)$ then $(a_1,a_3) \notin K^2(r_1,r_3)$.

	\item[S4] (Difference) For any $x,y,z$ with $y \neq z$, $(x,y) \in K^2(x,z)$ and  $(y,x) \in K^1(z,x)$.	
	
	\item[S5] (Diminishing Sensitivity) For any $x,y,K^1,K^2,\epsilon>0$,
	if $(x,y) \notin K^1(r_1,r_2)$, then $(x+\epsilon,y) \notin K^1(r_1+\epsilon,r_2)$.
	\item[S6] (Equal Salience) For any $x,r \in X$: if $\frac{x_1}{r_1}=\frac{x_2}{r_2}$ or $\frac{x_1}{r_1}=\frac{r_2}{x_2}$, then $x \notin K^k(r)$ for $k=1,2$.

\end{description}
\begin{figure}[h]
	\begin{center}
		\includegraphics[width=.9\textwidth]{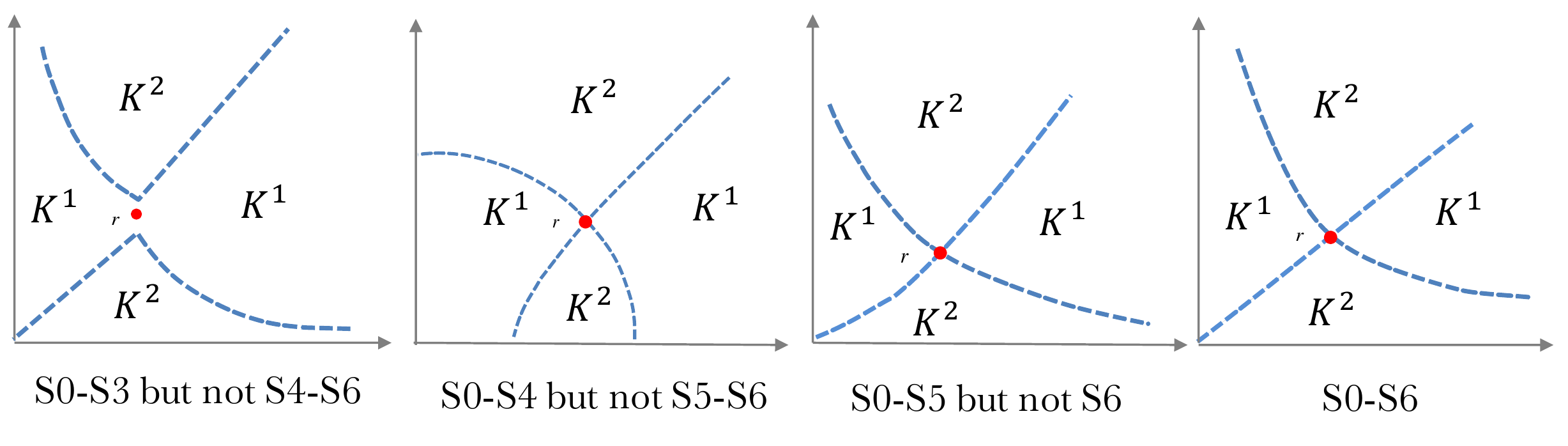}
	\end{center}
	\caption{Properties \textbf{S0-S6} Illustrated}
	\label{fig:categories}
\end{figure}
The properties have natural interpretations. Any category function satisfies \textbf{S0} by definition; we include it for completeness. {\bf{S1}} indicates that  making a bundle's less salient attribute closer to the reference point does not change the salience of the bundle. That is,  when $x$ and $y$ differ only in attribute $l$, and $y$ is closer to the reference in that attribute, if $x$ is $k$-salient, then so is $y$.
{\bf{S2}} requires that the same ranking is used for each attribute. 
{\bf{S3}} adapts transitivity to the salience ranking. It says that if $a_1$ stands out more relative to $r_1$ than $a_2$ does to $r_2$, and $a_2$ stands out more relative to $r_2$ than $a_3$ does to $r_3$, then $a_1$ stands out more relative to $r_1$ than $a_3$ does to $r_3$.
{\bf{S4}} says simply that any difference stands out more than no difference.
{\bf{S5}} implies that increasing both the good and the reference by the same amount in the same dimension does not move the good from one category to another.
{\bf{S6}} reads that if every attribute of $x$ differs from the reference point by the same percentage, then none of the attributes stands out.  More formally, if the percentage difference between $x_k$ and $r_k$ is the same across attributes, then $x$ is not $k$-salient for any $k \in \{0,1\}$.
%
%

Figure \ref{fig:categories} provides examples of categories that satisfy some but not all of the properties. Their formal definition and a verification that they satisfy the claimed properties can be found in Example \ref{ex: regions BGS} in the Appendix.

We say that \emph{categories are generated by a salience function} $\sigma$ if $x\in K^i(r)$ if and only if $\sigma(x_i,r_i) >\sigma(x_{j},r_j)$ for all $j \neq i$.
Thoerem \ref{thm: BGS regions} shows that category functions satisfying \textbf{S0-S4} are so generated.
\textbf{S5} and \textbf{S6} impose diminishing sensitivity and homogeneity of degree zero, respectively.
\begin{thm}
	\label{thm: BGS regions}
	The category function satisfies:\vspace{-1.5em}
	\begin{enumerate}
		\item \textbf{S0-S4} if and only if there exists a salience function  $\sigma$  that generates it;
		\item \textbf{S0-S5} if and only if the $\sigma$ that generates it has  diminishing sensitivity; and
		\item \textbf{S0}, \textbf{S1}, and \textbf{S6} if and only if it satisfies \textbf{S0-S6}  if and only if it is generated by an HOD salience function  $\sigma$. Any HOD salience function generates the same categories.
	\end{enumerate} 
\end{thm}

The result characterizes categories generated by a salience function.%
\footnote{Theorem \ref{thm: BGS regions} relies on the full structure of $\mathbb{R}^2$ for the last two results, as noted in Footnote \ref{fn: Rn}. Diminishing sensitivity and Homogeneity are both cardinal properties, and so are undefined without cardinal structure on $X$. Properties \textbf{S0-S4} are defined. Subsequent results that rely on Theorem \ref{thm: BGS regions}, such as Propositions \ref{thm:BGS preference} and \ref{thm: BGS choice}, remain true when imposing only \textbf{S0-S4} in this setting.}  It translates the functional form assumptions on the salience function in terms properties of categories. The most common specifications of the salience function are all  HOD, and so satisfy all of the above properties. Surprisingly, the result shows that there is a unique category function satisfying all the properties. Hence, any two HOD salience functions lead to exactly the same behavior.

We now turn to the question of identifying the salience of alternatives from choice behavior alone. 
\begin{proposition}
	\label{prop: revealing regions}
	Given that $\{\succsim_r\}_{r\in X}$ has a BGS representation, the categories are uniquely identified. The category function is $\hat{\mathcal{K}}=(\hat{K}^{1},\hat{K}^{2})$ with \[\hat{K}^{i}(r)=
int \left\{x \in X: \exists \epsilon>0\ s.t.\ \forall y \in B_{\epsilon}(x),\ y \sim_r x \iff y \sim_{r^i_x} x \right\}
\]
for every $r$, where $r^1_x=(x_1/2,x_2)$ and $r^2_x=(x_1,x_2/2)$.
\end{proposition}
Given a family $\{\succsim_r\}_{r\in X}$, the result identifies which alternatives have what salience.
As BGS necessarily satisfies the LIS condition, Proposition \ref{prop: general category ID} ensures that the categories are identified. This result improves on the previous one by providing an expression, solely in terms of the primitives, for the set of alternatives in each category. This ensures that the modeler can identify the categorization, and so the salience function, endogenously, i.e. from the DM's behavior alone.
This facilitates a full answer to the question of when a DM has a BGS representation for some salience function.

	In addition to the particular form of categories, BGS satisfies two properties that distinguish it from other CTMs. The most general of these is Reference Irrelevance, above, making BGS a Strong CTM. The other follows.
	
	\begin{ax}[Salient Dimension Overweighted, SDO] For any $x,y,r,r'\in X$:\\
		if $x,y \in K^k(r) \cap K^l(r^\prime)$, 
		$x \succsim_r y$, $x_l>y_l$, and $y_k>x_k$, then $x \succ_{r^\prime} y$.
	\end{ax}
	
	This axiom requires that categories correspond to the dimension that gets the most weight. That is, the DM is more willing to choose an alternative whose ``best'' attribute is $k$ when it is $k$-salient. 
	To illustrate, consider alternatives $x,y$ with $x_1>y_1$ and $y_2>x_2$. Because $x$ is relatively strong in attribute $1$,  $x$ should benefit more than $y$ from a focus on it. If $x$ is chosen over $y$ when attribute $2$ stands out for both, then this advantage in the first dimension is so strong that even a focus on the other one does not offset it. Hence, the DM should surely choose $x$ over $y$ for sure when attribute $1$ stands out for it.

	\begin{proposition} \label{thm:BGS preference} 
	Assume that there exist $x \in K^k(r)$ and $y  \in K^j(r)$ with $x \sim_{r} y$  for any categories $k,j$ and any $ r \in X$. Then, the family $\{\succsim_r\}_{r\in X}$ satisfies the Affine CTM axioms, Reference Interlocking, Reference Irrelevance, and Salient Dimension Overweighted for a category function  $\mathcal{K} $ satisfying   \textbf{S0}-\textbf{S5} if and only if it has a BGS representation where $\sigma$ has diminishing sensitivity.
	\end{proposition}

	This result characterizes the BGS model. It also provides guidance for comparing it with other models in the CTM class (see Figure \ref{fig:models} and Table \ref{table:comparison}). By outlining the model's testable implications, the result provides guidance on how to design experiments to test it.\footnote{The assumption that alternatives indifferent to each other exist in each category for each reference point is not strictly necessary. A sufficient condition for it to be necessary is that the utility indexes are both unbounded above (or below).} 
	
	\cite{bordalo2013salience} focus on a special case where the model is linear:  $w^1_1=w^2_2=1-w^2_1=1-w^1_2>\frac12$ and $u_1(x)=u_2(x)=x$. In an earlier version of this paper, we show this model is characterized by strengthening Affine Across Categories to require linearity and imposing a reflection axiom that requires permuting two alternatives and the reference point in the same way not to reverse the DM's choice between the two.\footnote{Formally, the first is that Affine Across Categories holds with $\oplus^k$ replaced by the usual $+$ operation. The second is that $(a,b) \succsim_{r_1,r_2} (c,d)$ if and only if $(b,a) \succsim_{r_2,r_1} (d,c)$. One can verify that these additional assumptions imply that the ancillary assumption about indifference holds.}

Taken together Propositions   \ref{prop: revealing regions} and \ref{thm:BGS preference} provide an outline for a fully subjective axiomatization of a family of preferences with a BGS representation. 
\begin{cor}
\label{cor: endog BGS}
Assume that there exist $x \in \hat{K}^k(r)$ and $y  \in \hat{K}^j(r)$ with $x \sim_{r} y$  for any categories $k,j$ and any $ r \in X$. Then, the family $\{\succsim_r\}_{r\in X}$ satisfies the Affine CTM axioms, Reference Interlocking, Reference Irrelevance, and Salient Dimension Overweighted for 
$\hat{\mathcal{K}}$ and
$\hat{\mathcal{K}}$ satisfies   \textbf{S0}-\textbf{S5} if and only if it has a BGS representation where $\sigma$ has diminishing sensitivity.
\end{cor}
We illustrate necessity of the result.
When the family of preferences has a BGS representation, Proposition \ref{prop: revealing regions} shows that the category function equals $\hat{K}$. Moreover, $\hat{K}$ satisfies S0-S5 when $\sigma$ has diminishing sensitivity by Theorem \ref{thm: BGS regions}, and the family of preferences satisfy the axioms in Proposition \ref{thm:BGS preference} for $\hat{K}$.
Thus, the result provides the complete testable implications of BGS in terms of the family of preferences alone.
\section{Choice Correspondence}
\label{sec: Choice Correspondence}

In this section, the modeler observes the DM's choice from finite subsets of alternatives but not her reference point. A model consists of both a theory of reference formation and a theory of choice given categorization.
In this setting, we can jointly test the theory of choice given categorization, categorization given reference, and reference formation.

We model reference formation via a \emph{reference generator} $A$ that maps finite subsets of alternatives to reference points, with the interpretation that $A(S)$ is the reference point when the menu is $S$.
Examples include the BGS theory that $A(S)$ is the average alternative, that $A(S)$ is the median bundle, that $A(S)$ is the upper (or lower) bound of $S$, and the \cite{KR06} theory that $A(S)=c(S)$.
If additional observable data on the choice context is provided, then it is easy to extend our results to $A$ being a function of that as well. For instance, \cite{MO05} theorize that the initial endowment $e$ is observable and that $A(S,e)=e$, and \cite{BGS_MAC} theorize that past histories $h$ of consumption are available and that $A(S,h)$ is the average between the bundles in $S$ and those in $h$.

Fixing a categorization function $\mathcal{K} $ and a reference generator $A$, let $\mathcal{X}$ be the set of finite and non-empty subsets of $X$ such that every alternative is categorized. Formally, $S\in \mathcal{X}$ only if  $ S \subset \bigcup_{i=1}^m K^i(A(S))$. We call these categorized menus or  menus for short. The requirement ensures that each alternative in the choice set belongs to a category given the reference point  $A(S)$. We leave open how the DM chooses when alternatives that are uncategorized belong to the choice set. By leaving the choice from this small set of menus ambiguous, we can more clearly state the properties of choice implied by the model.%
\footnote{One can, of course, extend the model to account for these choices. For instance, \cite{bordalo2013salience} hypothesize that these alternatives are evaluated according to their sum. Complications arise because the uncategorized alternatives are ``small:'' its complement is open and dense.}

We summarize the DM's choices by a choice correspondence $c:\mathcal{X} \rightrightarrows X$ with $c(S) \subseteq S$ and $c(S) \neq \emptyset$ for each $S \in \mathcal{X}$.
\begin{defn}
	The choice correspondence $c$  \emph{conforms to Strong-CTM under $(\mathcal{K} ,A)$} if there exists a family of preference relations $\{\succsim_r\}_{r\in X}$ that conforms to Increasing Strong CTM under $\mathcal{K} $
	so that $$c(S)= \left\{x \in S: x \succsim_{A(S)} y\text{ for all }y\in S \right\}$$ for every $S \in \mathcal{X}$.
\end{defn}

\subsection{Reference point formation}

Provided that the reference generator is responsive enough to changes in the menu, there is the possibility of testing the properties required by categorization on $\succsim_r$. One example of enough structure is that the reference point is the average bundle. However, this is just one example. An even more general sufficient condition is as follows.

\begin{ass*}
	A function $A$ is a \emph{generalized average} if for any $S=\{x^1,\dots,x^m\} \in \mathcal{X}$:\\
	(i) the function $x\mapsto A([S\setminus \{x_1\}] \bigcup \{x\})$ is continuous at $x_1$, and\\
	(ii) 
	for any $\epsilon>0$ and any finite $S' \subset \bigcup_i K^i(A(S))$, there exists $S^* \in \mathcal{X}$ so that $S^* \supset S \bigcup S'$, $d\left( A \left( S^*  \right) , A (S) \right)<\epsilon$, and for any $x' \in S^* \setminus S'$, $\min_{x\in S}d(x',x)<\epsilon^2$.
\end{ass*}

Examples of generalized average reference  include the average bundle $$A_{a}(S)=\left( \frac{\sum_{x\in S} x_1}{|S|},\frac{\sum_{x\in S} x_2}{|S|}\right),$$
the median value of each attribute, 
and a weighted average $$A_{wa}(S)=\left( \frac{\sum_{x\in S} w(x)x_1}{\sum_{x\in S} w(x)},\frac{\sum_{x\in S} w(x) x_2}{\sum_{x\in S} w(x)}\right)$$ for any continuous weight function  $w:X \rightarrow [a,b]$ with $b>a>0$.
We sometimes impose the additional requirement that $A(S) \in co(S) \setminus ext(S)$  for all non-singleton $S$; if so, we call $A$ a \emph{strong generalized average}.
The first and last of these examples satisfy this property.
The supremum and infimum are \emph{not} generalized averages, nor (necessarily) is the choice acclimating reference generator, $c(S)=A(S)$.\footnote{Recall $\sup S= \left( \max_{x \in S} x_1,\max_{x \in S} x_2 \right)$ and $\inf S$ is defined analogously.}

\subsection{Behavioral Foundations for Strong-CTM}

We now consider the behavior by a DM who conforms to Strong-CTM for a given category function and reference generator.
To do so, we make use of our earlier analysis by revealing how the DM evaluates alternatives categorized in a given way.
When $A(S)$ is a generalized average, this provides enough structure to identify enough of the family to apply our earlier analysis.

The main behavioral content comes from the choice correspondence equivalent of Reference Irrelevance. 
To state it, we introduce the following definition and notation.
\begin{defn}
	The alternative $x$ in category $k$ is indirectly revealed preferred to alternative $y$ in category $j$, written $(x,k) \succsim^R (y,j)$, if there exists finite sequences of pairs $(x^i,S^i)_{i=1}^{n}$  such that $x=x^1 \in K^k (A(S^1))$, $y \in K^j (A(S^n))\bigcap S^n$, and for each $i$: $x^i \in c(S^i)$, $x^{i+1} \in S^i$, and $x^{i+1} \in K^{k_i}(A(S^i)) \cap K^{k_i}(A(S^{i+1}))$ for some $k_i$.
\end{defn}
To interpret the definition, consider menus $S^1,S^2$ and alternatives $x^1 \in  c(S^1)$ and $x^2 \in c(S^2) \cap S^1$, where $x^2$ is categorized in the same way for both menus. For example, $x^1$ is in category $1$ for $S^1$, and $x^2$ is in category  $2$ for both. The observation $x^1 \in  c(S^1)$ reveals that the valuation of $x^1$ is at least as high as that of $x^2$ when $x^1$ belongs to the first category and $x^2$ to the second. Pick any $y \in S^2$, and suppose the DM categorizes it as $j$ in $S^2$. Since  $x^2$ is chosen from $S^2$, the DM perceives that $x^2$ has a higher value than  $y$, when she categorizes the first as $2$ and the second as $j$. By transitivity, the DM also perceives that $x^1$ in $1$ has a higher value than $y$ in $j$. The relations $\succsim^R$ captures this and extends it to longer sequences. 

We replace Reference Irrelevance with the following weakening of the Strong Axiom of Revealed Preference (SARP).

\begin{ax1}[Category SARP]\label{ax:firstc}\label{ax:SSARP}
	For any $S \in \mathcal{X}$, if $(x,k) \succsim^R (y,j)$, $x \in K^k(A(S)) \bigcap S$, $y \in K^j(A(S)) \bigcap S$, and $y \in c(S)$, then $x \in c(S)$.\end{ax1}

Consider menus $S^1,S^2$ that both contain $x^1$ and $x^2$  where $x^2$ is in category  $k$ for both $S^1$ and $S^2$.  If $x^1$ is chosen from and belongs to category $1$ for $S^1$, then $(x^1,1) \succsim^R (x^2,k)$.  If  $x^1$ belongs to category $1$ in $S_2$, then  the DM values both it and $x^2$ the same in $S^1$ as in $S^2$ since neither's categorization  changed. If $x^2$ is chosen from $S^2$, then $x^1$ must be chosen as well. In particular, the DM obeys the Weak Axiom of Revealed Preference (WARP) whenever she categorizes chosen alternatives in the same way. However, the axiom leaves open the possibility of a WARP violation when either is differentially categorized.


The axiom extends this logic to sequences of choices in much the same way that SARP does to WARP. A finite sequence of choices, where the  choice from the next menu is available in the current one and has the same salience in both, does not lead to a choice reversal. Since salience does not change along the sequence of choices, the choices do not exhibit a reversal.

Category SARP limits the effect of unchosen alternatives.
Modifying them can alter the DM's choice, but only insofar as it changes the reference point and thus the salience of alternatives.
When comparing the same two alternatives in different menus, the DM's relative ranking does not change when neither's salience changes.
This property greatly limits the effect of the reference point.   In fact, a sufficiently small change in the reference never leads to a preference reversal.

The remaining axioms are the natural generalizations to the choice correspondence of Category Cancellation, Category Monotonicity, Category Continuity, Reference Interlocking, and Affine Across Categories.
We denote these by appending a ``*'' to distinguish from their reference-dependent-preference formulation.
Appendix \ref{sec: axioms for c} contains their formal statement.

As before, we require some additional topological structure on the categories.
For a category $k$, let $$E^{R,k} = \{x \in X: x\in K^k(A(S)), \ \{x\}= c(S) \}$$  and
$$D^k=\bigcup_{S \in \mathcal{X}} \left\{K^k(A(S)) \bigcap S \right\} .$$ 
The generalization of the structure assumption is as follows.
\begin{ass*}[Revealed Structure]
	For any category $k$, $E^{R,k}$ is open, $E^{R,k}$ is dense in $D^k$, and the following sets are connected: $E^{R,k}$, $\{x\in E^{R,k}: x_j=s\}$ for all dimensions $j$ and scalars $s \in \mathbb{R}$, and 
	$\{y \in E^{R,k}: (x,k) \sim^{R} (y,k) \}$ for all $x\in E^{R,k}$.
\end{ass*}
In addition to what was imposed by the Structure Assumption, we require that almost all objects categorized in a category are chosen in some menu.
This can be weakened, but is typically satisfied by the models in which we are interested, such as BGS.

We require one last assumption.
\begin{ax1}[Comparability Across Regions, CAR]
	\label{ax: BGS-c-cont}
	If $x \in E^{R,k}$, then for any $j$ there exists 
	$y\in E^{R,j}$ so that $(x,k) \sim^{R} (y,j)$.
\end{ax1}
This is a  version of the assumption in Theorem \ref{result: strong CTM}.
It requires that every alternative chosen when it belongs to category $k$ is revealed to be equally good to some other alternative when it is categorized in category $j$. 
With it, we can now state the result.
%
%


\begin{thm}
	\label{thm: strong CTM choice}
	Assume that Revealed Structure 	and CAR hold and that $A$ is a generalized average. 
	A choice correspondence $c$ conforms to strong-CTM under $(\mathcal{K},A) $ if and only if $c$ satisfies Category-SARP, Category Monotonicity*, Category Cancellation*, Category Continuity*, and Affine Across Categories*.	
\end{thm}

The result is the counterpart of Theorem \ref{result: strong CTM} with an endogenous reference point.
The behavior corresponding to categorization does not fundamentally change across settings.
As long as the DM reacts consistently when alternatives are categorized in the same way, then we can represent her choices as categorical thinking where the reference point only affects how she categorizes each alternative.
The key challenge in the proof is to establish that the arguments we used to establish our earlier results still hold.
We adapt our earlier arguments to show that revealed preference within category $k$ is complete on $E^{R,k}$. This relies on small changes in alternatives not changing choice, a property implied by generalized average. Then, the remaining axioms establish that this within-category preference has an additive representation. CAR allows us to extend across categories.

\subsection{Behavioral Foundations for BGS}
In this subsection, we provide a behavioral foundation for BGS.
The first step is to show that the Revealed Structure assumption holds.
\begin{lemma}\label{lem:BGS region}
	If $A$ is a strong generalized average, $\mathcal{K} $ satisfies  \textbf{S0},  \textbf{S1}, and  \textbf{S4}, and $c$ satisfies Category Montonicity*, then $E^{R,k}=\mathbb{R}^2_{++}$ for $k=1,2$.
\end{lemma}
Given the assumptions we have made so far, every alternative is chosen in some menu when it is $k$-salient. Consequently, the revealed structure assumption must hold. The result relies  on the observation that the DM categorizes $x$ as $1$-salient when all other available options have the same value in dimension $2$ as $x$. If $x$ has the highest value in attribute $1$ in such a choice set, then it must be chosen. 

Now, we can apply Theorem \ref{thm: strong CTM choice} in combination with the insights gained from Proposition \ref{thm:BGS preference} to understand the behavioral foundation of the BGS model.

\begin{proposition}
	\label{thm: BGS choice}
	Assume that $A$ is a strong generalized average and that  CAR holds.
	The choice correspondence $c$ satisfies   Category-SARP, Category Monotonicity*, Category Cancellation*, Category Continuity*, Affine Across Categories*, Reference Interlocking*, and Salient Dimension Overweighted* for a category function  $\mathcal{K} $ satisfying   \textbf{S0}-\textbf{S5} if and only if $c$ conforms to BGS  where $\sigma$ has diminishing sensitivity.
\end{proposition}

This proposition lays out the behavioral postulates that characterize the BGS model with endogenous reference point formation. Most importantly, it connects the (unobserved) components of the model to observed choice behavior. Fundamentally, the properties that Proposition \ref{thm:BGS preference}  characterized the model in our first setting still characterize it. To do so, we note that Theorems \ref{thm: BGS regions} and \ref{thm: strong CTM choice} imply that there exists a Strong CTM with categories generated by a salience function. We then establish that choice within the $k$-salient alternatives overweights dimension $k$ by using SDO and the structure of regions.

Finally, we ask the question of whether the choice correspondence with an endogenous reference point provides enough leverage to identify salience.
	\begin{proposition}
		\label{prop: revealing regions with c}
		Given that $c$ conforms to BGS with a strong generalized average, the categories are uniquely identified.
	\end{proposition}

As with Propositions \ref{prop: revealing regions} and \ref{thm:BGS preference}, Propositions \ref{thm: BGS choice} and \ref{prop: revealing regions with c} provide a roadmap for testing BGS without a known salience function. However, it still requires that the reference generator is a strong generalized average.
Consequently, the axioms capture the full testable implication of the model and allow for tight comparisons with other existing work.

\section{Related Literature}
\label{sec: conclusion}

This paper  is closely  related to the literature which studies how a reference point affects choices, (e.g. \citet*{TK91,munro2003theory,sugden2003reference,MO05,sagi2006anchored,
salant2008f,apesteguia2009theory,masatlioglu2013choice,MO2014,dean2017limited}). 
The papers focus on an exogenous reference point, as in Section \ref{sec: axiomatization}.
While TK and MO are examples of CTM, the others are not. Nonetheless, our analysis puts the models on an equal footing so their implications can be compared.

Our extension to endogenous reference point formation adopts the approach of a number of  recent papers, e.g. \citet*{bodner1994centroid,kivetz2004alternative,orhun2009,bordalo2012salience,
tserenjigmid2015choosing}. As in Section \ref{sec: Choice Correspondence}, the reference point is a function of the context, and is identical for all feasible alternatives. Finally, \citet*{KR06,ok2014revealed,freeman2017preferred} and \cite{kibris2018} study models where the endogenous reference point is determined by what the agent chooses, but is otherwise independent of the choice set. This represents a very different approach to reference formation, and our approach does not easily generalize to accommodate it.\footnote{\citet*{maltz2017} is the only model of which we are aware that combines an exogenous reference point with endogenous reference-point formation.} 

One of our key contributions is to provide an axiomatization of the salient-thinking model. 
Recent work by \cite{lanzani2019} introduces a model of risk preferences where the correlation between outcomes affects the pair-wise ranking of monetary lotteries. The salient-thinking model under risk is a special case  and an axiomatic characterization is provided. Other than the domain, a key distinction with our result is that the DM violates transitivity, which we avoid by considering reference-dependent preferences. 

Interpreting salience as arising from  differential attention to attributes, CTM has a close relationship with the literature studying how limited attention affects decision making.
 \cite{MNO12} and \cite{manzini2014stochastic}  study a DM who has limited attention to the alternatives available. The DM maximizes a fixed preference relation over the consideration set, a subset of the alternatives actually available. In contrast, in CTM the DM the considers all available alternatives but maximizes a preference relation distorted by her attention.  \cite{caplindean2013}, \cite{deOliveraetal2016} and \cite{ellis2013} study a DM who has limited attention to information. In contrast to CTM, attention is chosen rationally to maximize ex ante utility, rather than determined by the framing of the decision, and choice varies across states of the world. The most related interpretation considers attributes as payoffs in a fixed state. In addition to choices varying across states, each alternative has the same weights on each attribute, similar to \cite{koszegi2013model}. 
Taken together, these results highlight the effects on behavior of different types of attention.

While we argue in this paper that a number of prominent behavioral economic models can be thought of as resulting from categorization, few papers in economics explicitly address categorization.
 \cite{mullainathan2002thinking} provides a model of belief updating and shows how categorization can generate non-Bayesian effects. \cite{fryer2008categorical} introduce a categorical model of cognition where a decision maker categorizes her past experiences. Since the number of categories is limited, the decision maker must group distinct experiences in the same category. In this model, prediction is based on the prototype from the category which matches closely the current situation. 
Finally, \cite{manzini2012categorize} introduce a two-stage  decision-making model. In the first stage, a decision maker eliminates some alternatives based on the categories to which they belong, and in the second stage she maximizes her preference among those that  survived the  first stage. 
\cite{BGS_MAC} provide a model of memory and attention, where the context's similarity to past consumption opportunities affects the salience of the alternatives currently available. They show this leads to endogenous categorization of the current opportunity set, and discuss the resulting implications for choice.


\appendix

\singlespacing
\section{Proofs and Extras from Sections \ref{sec: Model} - \ref{sec: BGS}}
\subsection{Proof of Theorem \ref{thm: weak CTM}}
\label{pf: thm: weak CTM}
\begin{lemma}
	 $\succ^{k*}$ has open upper and lower contour sets in $E^k$.
	 \label{lem: cont trans closure}
\end{lemma}
\begin{proof}
	Suppose $x\succ^{k*} y$. Then, there are $x^1,x^2,\dots, x^M \in E^k$ and $r^1,\dots,r^{M-1}$ with $x^1=x$ and $x^M=y$ so that $x^j\succsim_{r^j} x^{j+1}$ and $x^j,x^{j+1} \in K^k(r^j)$.
	Let $\epsilon_j>0$ be such that $B_{\epsilon_j}(x^j),B_{\epsilon_j}(x^{j+1}) \subset K^k(r^j)$.
	Set $\epsilon=\min\{ \epsilon_j\}_{j<M} $.
	
	Now, $x^j \succ^k x^{j+1}$ (and so $x^j \succ_{r^j} x^{j+1}$) for at least one $j$. Let $m$   be an index for which this is true.
	Since $B_{\epsilon}(x^{m}),B_{\epsilon}(x^{m+1}) \subset K^k(r^m)$,
	there exists $0<\epsilon^*_m<\epsilon$ be such that $B_{2\epsilon^*_m}(x^m)$  is a subset of $\{x \in K^k(r^m):x \succ_{r^m} x^{m+1}\}$ by Category Continuity.
	Then, $x^m-\epsilon^*_m \succ_{r^m} x^{m+1}$, so $x^m-\epsilon^*_m \succ^k x^{m+1}$ and, by definition of $\succsim^{k*}$, it follows that $x^m-\epsilon^*_m \succ^{k*} y$. Assume (IH) that there is $\epsilon^*_{m-j}\in (0, \epsilon)$ so that $x^{m-j}-\epsilon^*_{m-j} \succ^{k*} y$. 
	Then,   \[
	x^{m-j-1} \succ_{r^{m-j-1}} x^{m-j}-\epsilon^*_{m-j} \]
	since  $ x^{m-j} \succ_{r^{m-j-1}} x^{m-j}-\epsilon^*_{m-j}$ by Category Monotonicity, $x^{m-j-1} \succsim_{r^{m-j-1}} x^{m-j}$ by definition, and transitivity of $\succsim_{r^{m-j-1}}$.
	By Category Continuity and Monotonicity, there then exists $\epsilon^*_{m-j-1} \in (0,\epsilon)$ so that $x^{m-j-1}-\epsilon^*_{m-j-1} \succ_{r^{m-j-1}} x^{m-j}-\epsilon^*_{m-j} $, and by definition it follows that $x^{m-j-1}-\epsilon^*_{m-j-1} \succ^k x^{m-j}-\epsilon^*_{m-j} $.
	By (IH), Weak Reference Irrelevance, and the definition of $\succsim^{k*}$, it follows that $x^{m-j-1}-\epsilon^*_{m-j-1} \succ^{k*} y$.  
		Therefore, there is $\epsilon^*_1 \in (0, \epsilon)$ so  that $ x^1 - \epsilon^*_1 \succ^{k*} y$, so by Category Monotonicity, Weak Reference Irrelevance, and definition of $\succsim^{k*}$, we have $x' \succ^{k*} y$ for any $x'\in B_{\epsilon^*_1}(x)$. Conclude the upper-contour set is open; similar arguments hold for the lower-contour set.	
\end{proof}

\begin{lemma}
	\label{lem: trans closure complete}
	$\succsim^{k*}$ is complete on $E^k$.
\end{lemma} 
\begin{proof}
	Pick any $x,y\in E^k $ and let $E^*=E^k \bigcap B_{d(x,y)+1}(x)$.
	As the intersection of two intersecting connected sets, $E^*$ is connected, and as a subset of $\mathbb{R}^n$, there is a continuous path $\theta:[0,1]\rightarrow E^*$ so that $\theta(0)=x$ and $\theta(1)=y$.
	
	This $\theta$ can be chosen so that it crosses each $\succsim^{k*}$ indifference curve at most once. To see why, suppose that $\theta(a) \sim^{k*}  \theta(b)$ and $b>a$.
	First, we show that $IC=\{b' \in E^*: b' \sim^{k*}  \theta(a) \}$ is path-connected. Then, $IC_{-n}=\{y_{-n} :y  \in IC \} \subset \mathbb{R}^{n-1} $ is also connected as the projection of $IC$ onto the first $n-1$ coordinates.
Moreover it is open since $y \in E^k$ implies there is a reference $r$ and $\epsilon>0$ so that $B_\epsilon(y) \subset K^k(r)$ and $\succsim_r$ is complete, transitive, monotone, and continuous when restricted to $B_\epsilon(y)$.
Conclude $IC_{-n}$ is path-connected as a connected open subset of $\mathbb{R}^{n-1}$. 
Now for any $a',b' \in IC$, there is a path $\theta''_{-n}:[0,1]\rightarrow IC_{-n}$ from $a'_{-n}$ to $b'_{-n}$, and 
for each $c_{-n} \in I$ there is a unique $c_n$ so that $c=(c_{-n},c_n) \in IC$ by category monotonicity.
Let $\theta''_n(x)$ be such that $(\theta''_{-n}(x),\theta''_n(x)) \in IC$. $\theta''_n$ is continuous since $IC$ is closed in $E^k$. Then, $\theta''=x \in [0,1] \mapsto (\theta''_{-n}(x),\theta''_n(x)) \in IC$ is the desired path.
Hence there is a path $\theta':[0,1]\rightarrow IC$ with $\theta'(0)=\theta(a)$ and $\theta'(1)=\theta(b)$.
	Then the path $\theta^*$ given by $\theta^*(x)=\theta(x)$ for $x\notin [a,b]$ and $\theta^*(x)=\theta' \left( \frac{x-a}{b-a} \right)$ for $x\in [a,b]$ is also a continuous path from $x$ to $y$. Constructing this for $a^* = \min \{a':\theta(a') \sim^{k*} \theta(a) \} $ and $b^* = \max \{a':\theta(a') \sim^{k*} \theta(a)\}$ gives a path that crosses $IC$ at most once. These are well-defined since $\theta$ is continuous.
	
	Now, let $Y=\theta^{-1}([0,1])$. $Y$ is closed since $\theta$ is continuous and so compact as a subset of $cl(B_{d(x,y)+1}(x))$. For any $z\in Y$, there exists $r_z \in X$ and $\epsilon_z>0$ so that $B_z= B_{\epsilon_z}(z) \subset K^k(r_z)$.	Since $B_z\subset  K^k(r_z)$ and $\succsim^k$ is a subrelation of   $\succsim^{k*}$,  $\succsim^{k*}$ is complete and transitive when restricted to $B_z$.
	Then, the collection $\{B_z:z\in Y\}$ is an open cover of $Y$ and hence has a finite subcover $B_{z_1},B_{z_2},\dots, B_{z_m}$. W.L.O.G., $B_{z_j}$ is not a subset of $B_{z_{j'}}$ for any $j,j'$ and $\theta(z_j) < \theta(z_{j+1})$, so $x \in B_{z_1}$ and $y \in B_{z_m}$.
	Moreover, since $\theta$ crosses each indifference curve only once, if $z_k \succ^{k*} z_{k+1}$ ($z_k \prec^{k*} z_{k+1}$) for any $k$, then $z_j \succsim^{k*} z_{j'}$ ($z_k \precsim^{k*}   z_{k+1}$) for any $j'>j$.
	W.L.O.G. consider the former. Pick $a^1 \in B_{z_1} \bigcap B_{z_2} \bigcap Y$ so that $x\succsim^k a^1 $ and then pick $a^j \in B_{z^j} \bigcap B_{z^{j+1}} \bigcap Y$ so that $a^{j-1}\succsim^k a^j$. Then, $$x \succsim^{k*} a^1 \succsim^{k*} a^2 \succsim^* \dots \succsim^{k*} a^m \succsim^{k*} y.$$ Since $\succsim^{k*}$ is transitive, we conclude $x \succsim^{k*} y$. Since $x,y$ were arbitrary, $\succsim^{k*}$ is complete.
\end{proof}

Apply CW Theorem 2.2 to get an additive representation $U^i(x)$ on $E^i$.
For any $x,y \in K^i(r)$, $x \succsim_r y$ if and only if $U^i(x) \geq U^i(y)$ and $U^i(x)=\sum_j U^i_j(x_j)$.

\begin{lemma}
	\label{lem: indif point or all same}
	For categories $K^i(r)$ and $K^j(r)$, either (i) there exists $ x^i \in K^i(r)$ and $x^j \in K^j(r)$ so that $x^i \sim_r x^j$; or (ii) $x^i \succ_r x^j$ for all $x^i\in K^i(r)$ and $x^j \in K^j(r)$; or (iii)  $x^j \succ_r x^i$ for all $x^i\in K^i(r)$ and $x^j \in K^j(r)$.
\end{lemma}

\begin{proof}
	If neither (ii) nor (iii) holds, then after relabeling categories if necessary, there exist $x \in K^i(r)$ and $y,z\in K^j(r)$ such that $y\succ_r x \succ_r z$. Let $UC_j (x)$ and $LC_j (x)$ be the strict upper and lower contour sets of $x$ in category $j$ for reference $r$.  Any point in $K^j(r) \setminus [UC_j (x) \bigcup LC_j (x)]$ is indifferent to $x$, so either (i) holds or the set is empty.
	There exists an $\epsilon>0$ such that for every $x' \in B_\epsilon(x)$,  $y\succ_r x' \succ_r z$ by Category Continuity and hence $K^j(r) \neq U_j(x')$ and $K^j(r) \neq L_j(x')$. 
	By Category Continuity, there exists $x'\in B_\epsilon(x)$  such that  $K^j(r) \setminus [UC_j (x') \bigcup LC_j (x')] \neq \emptyset$ (otherwise, $B_\epsilon(x)$ is contained in the interior of the set considered), so we can take $y' \in K^j(r) \setminus [UC_j (x') \bigcup LC_j (x')]$ and conclude $y' \sim_r x'$.
\end{proof}

\begin{defn}
	\label{def: IS}
	A finite sequence $(Q_{1},\dots,Q_{m+1})$ with each $Q_i \in \{ K^1(r), \dots ,K^n(r)\}$ is an \emph{indifference sequence for $r$} (IS) if there exists $x^1,\dots ,x^m,y^1,\dots,y^m$ with $x^k \in Q_{k}$, $y^k \in Q_{k+1}$ and $x^k \sim_r y^k$.
\end{defn}
We omit the dependence on $r$ when clear from context.

Define the relation $\bowtie_r$ by $x \bowtie_r y$ if there exists  an indifference sequence of categories $(Q_1,\dots,Q_m)$ with $x \in Q_1$ and $y \in Q_m$.
It is easy to see that $\bowtie_r$ is an equivalence relation (reflexive, symmetric, and transitive).
Let $[x]_r$ denote the $\bowtie_r$ equivalence class of $x$.
\begin{lemma}
	\label{lem: ordering of ISes}
	If $y \notin [x]_r$ and $x \succ_r y$, then $x^\prime \succ_r y^\prime$ for all $x^\prime \in [x]_r$ and $y^\prime \in [y]_r$.
\end{lemma}
\begin{proof} 
	Fix $x,y,r \in X$ with $y \notin [x]_r$ and $x \succ_r y$, and assume $x\in K^k$.
	Pick any $y^\prime \in [y]_r$.
	By definition, there is an IS $(Q_1,\dots,Q_m)$ with $y^\prime \in Q_m$ and $y \in Q_1$.
	Let $i=1$ and $y_1=y$. If there exists $y^{\prime\prime} \in Q_i$ with $y^{\prime\prime} \succsim_r x$, then $y^{\prime\prime} \succsim_r x \succ_r y_i$, so by  Lemma \ref{lem: indif point or all same}, we can find $z\in Q_i$ and $x' \in K^k$ with $z \sim_r x'$. If that occurs, then $(K^k, Q_i,\dots,Q_1)$ is an IS and $y \in [x]_r$, a contradiction. Thus $ x \succ_r y^{\prime\prime}$ for all $y^{\prime\prime} \in Q_i$.	
	Now, there exists $y_{i+1} \in Q_{i+1}$ with $x \succ_r y_{i+1}$ by transitivity and definition of IS. Hence, we can apply above logic to $Q_{i+1}$ as well:  $ x \succ_r y^{\prime\prime}$ for all $y^{\prime\prime} \in Q_{i+1}$. Inductively, this extends all the way to $Q_m$, so $x \succ_r y'$ in particular. Since $y'$ is arbitrary, this extends to any $y' \in [y]_r$.
	
	Similar arguments show that $x'\succ_r y$ for any $x^\prime \in [x]_r$.  Combining, $x' \succ_r y'$ whenever $x' \in [x]_r$  and $y' \in [y]_r$.
\end{proof}

	Fix a reference point $r$. Let $A_1,\dots,A_n$ be the distinct equivalence classes of $\bowtie_r$.
	By Lemma \ref{lem: ordering of ISes}, these sets can be completely ordered by $\succ_r$, i.e. $A_i \succ_r A_j \iff x \succ_r y$ for all $x \in A_i$ and $y\in A_j$.
	Label so that $A_1 \succ_r A_2 \succ_r \dots \succ_r A_n$.

	Pick an indifference class $A_i$ and an IS $Q_1,\dots,Q_M$ that contains points in every region in $A_i$.	
	We define $V_i(\cdot)$ on $A_i$ as follows.
	Define $V_i(x)$ on $Q_1$ so that $V_i(x)=U^{j}(x)$ for all $x \in K^j(r)$ where $K^j(r)=Q_1$. Clearly $V_i$ represents $\succ_r$ 
	when restricted to $Q_1$. There is no loss in assuming that $V_i$ is bounded, and the closure of its range is an interval.%
	\footnote{We can define $V'(x)=h(V(x))$ for $h(v)=-1/(1+v)$ when $v \geq 0$ and $h(v)=-2+1/(1-v)$ when $v<0$.}

	Now, assume inductively that,  for a given $m \leq k$, $V_i$ represents $\succ_r$ when restricted to $\bigcup_{j=1}^{m-1} Q_j\equiv \mathbf{Q}^{m-1}$, is bounded, is continuous on $\mathbf{Q}^{m-1}$, and is an increasing transformation of   $U^{k}$ within $Q_j$ when $Q_j=K^k(r)$.
	Then, extend $V_i$ to $Q_m$ as follows. 
	By Lemma \ref{lem: ordering of ISes}, it is impossible that $y \succ_r x$ for every $x \in \mathbf{Q}^{m-1}$ and every $y\in Q_m$.
	It will be convenient to relabel regions so that $Q_m=K^m(r)$.
	
	Pick a bounded, strictly increasing, continuous $h:\mathbb{R}\rightarrow \mathbb{R}$.
	For any $x \in K^m(r)$ so that $x \succ_r y$ for all $y \in \mathbf{Q}^{m-1}$,
	 set $$V_i(x)=h(U^{m}(x))+\beta_+$$ where 
	 	 $$\beta_+= \sup\{ V_i(x) :  {x \in \mathbf{Q}^{m-1}} \} - \inf \{h(U^{m}(x)):x\in K^m(r),\ x \succ_r y \text{ for all } y \in \mathbf{Q}^{m-1}\}.$$
	For any $x \in K^m(r)$ for which there exists $y,y' \in \mathbf{Q}^{m-1}$ so that $y\succ_r x \succ_r y'$, let $$V_i(x) = \inf \{V_i(y):y\in \mathbf{Q}^{m-1}\text{ and }y \succsim_r x \}.$$
	 For all other $x \in K^m(r)$, 
	 let $$V_i(x)=h(U^{m}(x))+\beta_-$$ where 
	 $$\beta_-= \inf\{V_i(x):x \in \mathbf{Q}^{m-1}\}-\sup \{h(U^{m}(x)):x\in K^m(r),\ y \succ_r x  \text{ for all } y \in \mathbf{Q}^{m-1}\}.$$
	This $V_i$  is bounded and continuous.
	
	We now show that it represents $\succ_r$ on $\mathbf{Q}^{m}$.
	Pick $x,y \in \mathbf{Q}^{m}$.
	There are four cases:\\
	\textbf{Case 1:} $x,y \in \mathbf{Q}^{m-1}$: then the claim follows by hypothesis.\\
	\textbf{Case 2:} $x \in K^m(r)$ and either $x \succ_r y'$ for all $y' \in \mathbf{Q}^{m-1}$ or $y' \succ_r x$ for all $y' \in \mathbf{Q}^{m-1}$: the claim is immediate.\\
	\textbf{Case 3:} $x \in K^m(r)$ and $y \in \mathbf{Q}^{m-1}$:
	If $y \succ_r x$, then $y-\epsilon \succ_r x$ for some $\epsilon>0$ so that $y-\epsilon$ belongs to the same region as $y$.
	If $y \sim_r x$, then $V_i(y) \geq V_i(x)$. If this does not hold with equality, then there is a $y' \in \mathbf{Q}^{m-1}$ so that $y' \succsim_r x$ and  $y \succ_r y'$ (since $y' \not\succsim_r y$). But then $y \succ_r x$, a contradiction.
	If $x \succ_r y$ but $V_i(y)\geq V_i(x)$, there exists $z \in \mathbf{Q}^{m-1}$ so that $V_i(z)\leq V_i(y)$ and $z \succsim_r x$.  But then by transitivity and hypothesis, $y \succsim_r z \succsim_r x$.\\
	\textbf{Case 4:} $x,y \in K^m(r)$ and Case 2 does not hold for either $x$ or $y$: Suppose $x\succsim_r y$. If not, then $V_i(y)>V_i(x)$ so there exists a $z \in \mathbf{Q}^{m-1}$ so that $z \succsim_r x$ and $z \not \succsim_r y$. By weak order, $y \succ_r z$ and so $y \succ_r x$, a contradiction.
	
	Since it represents $\succsim_r$ on $K^m(r)$, it also agrees with  $\succsim_m$ on $K^m(r)$. Hence it is an increasing transformation of   $U^{i}$ within $K^i(r)$ for each $i\leq m$.
	Renormalize $V_i$ so that its range is a subset of $[-\frac12 -i,-i]$.
	
	 For any $x,y \in A_i$, the above establishes that $V_i(x) \geq V_i(y) \iff x \succsim_r y$. For any $x\in A_i$ and $y\in A_j$ where $i<j$, $x \succ_r y$ by Lemma \ref{lem: ordering of ISes} and construction. Since $V_i(x)>-\frac12-i$, $V_j(y)<-j$, and $-\frac12-i>-j$,  we have $V_i(x)>V_j(y)$.
	Define $U^k(\cdot|r)$ to agree with the appropriate restriction of $V_i$, and conclude $\{\succ_r\}_{r\in X}$ conforms to CTM under $\mathcal{K}$.	
Since $r$ was arbitrary, this completes the proof. \hfill \qedsymbol
\subsection{Proof for  Theorem \ref{thm: RI}}
\label{pf: thm: RI}
Sufficiency is easy to verify. 
Suppose that $U^k(x)=\sum_{i=1}^n U^k_i(x_i)$. We show that for every category $j$ there exists a vector $w \gg 0$ so that $U^j(x)=\sum_{i=1}^n w_i U^k_i(x_i)$ represents $\succ_j$ on $E^k \bigcap E^j$.

Consider dimension 1, and the rest follow the same arguments. The goal is to show that $U^k_1(x)-U^k_1(y)\geq U^k_1(a)-U^k_1(b)$ if and only if $U^j_1(x)-U^j_1(y)\geq U^j_1(a)-U^j_1(b)$ for any $x,y,a,b \in E^k_1\bigcap E^j_1$.
If this is the case, then standard uniqueness results give that $U^j_1(x)=\alpha U^k_1(x)+\beta$. The $\beta$ can be dropped completing the claim.

Let $\pi_i$ be the projection onto the $i$-coordinate. Then, $E^k_1=\pi_1( E^k)$ is open and connected for any category $k$. 
This follows from $E^k$ connected and open and $\pi_i$ continuous. In $\mathbb{R}$, connected implies convex.

\begin{claim}
	For any $z \in E^k_1 \bigcap E^j_1$, there exists a neighborhood $O_z=B_{\epsilon_z}(z)$ so that $U^k_1(x)-U^k_1(y)\geq U^k_1(a)-U^k_1(b)$ if and only if $U^j_1(x)-U^j_1(y)\geq U^j_1(a)-U^j_1(b)$ for any $x,y,a,b \in O_z$.	
\end{claim}

To see it is true, pick $x\in E^k_1 \bigcap E^j_1$. Then there is an $a^l \in E^l$ with $a^l_1=x$ for $l=k,j$.
Let $U_{-i}^k(y)= \sum_{j\neq i } U^k_j(y_j)$ for any $y \in X$.
Since each $a^l \in K^l(r^l)$ for some $r^l \in X$, there exists an $\epsilon^l>0$ so that $B_{2\epsilon^l}(a^l) \subset K^l(r^l) \subset E^l$, where the distance is given by the supnorm.
Pick $\epsilon\in(0,\epsilon^l)$ so that $$U^l_1(x+\epsilon)-U^l_1(x-\epsilon)< U^l_{-1}(a^l+\epsilon^l)-U^l_{-1}(a^l-\epsilon^l)$$
for $l=k,j$.
Then, for any $a,b \in [x-\epsilon,x+\epsilon]$ there exists $y^a_{-1},y^b_{-1}$ 
so that $(a,y^a_{-1}), (b,y^b_{-1})\in B_{2\epsilon^k}(a^k)$ and $(a,y^a_{-1}) \sim_{r^k} (b,y^b_{-1})$ by Category Continuity and CM.
In particular, $U_1^k(a)-U_1^k(b)=U^k_{-1}(y^b_{-1})-U^k_{-1}(y^a_{-1})$.
For any $a',b'\in [x-\epsilon,x+\epsilon]$, it holds that $U_1^k(a)-U_1^k(b)\geq U_1^k(a')-U_1^k(b')$
if and only if $(b',y^a_{-1}) \succsim_{r^k} (a',y^b_{-1})$.
Similarly, there exist $z^a_{-1},z^b_{-1}$ so that $(a,z^a_{-1}), (b,z^b_{-1})\in B_{2\epsilon^j}(a^j)$ and $(a,z^a_{-1}) \sim_{r^j} (b,z^b_{-1})$. Now, 
$ (b',z^b_{-1}) \succsim_{r^j} (a',z^a_{-1})$ if and only if $U_1^j(a)-U_1^j(b)\geq U_1^j(a')-U_1^j(b')$. 
By Reference Interlocking and weak order, $ (b',z^b_{-1}) \succsim_{r^j} (a',z^a_{-1})$ if and only if $(b',y^a_{-1}) \succsim_{r^k} (a',y^b_{-1})$, so we conclude that the claim holds with $\epsilon_x=\epsilon$.

We now extend to the entire domain (this follows similar arguments in CW). 
Pick an arbitrary $x_*<x^* \in E^k_1 \bigcap E^j_1$ and consider $Z=(x_* ,x^*]$.
If the claim is true, then standard uniqueness results give that $U^j_1(x)=\alpha U^k_1(x)+\beta$ for all $x \in O_z$ for some $\alpha>0$.
Let $\alpha^*,\beta^*$ be the constants so that $U^j_1(x)=\alpha^*U^k_1(x)+\beta^*$ for all $x$ in the neighborhood of $x^*$, as guaranteed to exist by the claim.

Let $$Z_1=\left\{ s\in Z:U^j_1(x)=\alpha^*U^k_1(x)+\beta^* \text{ for all } x\in (x_* ,s] \right\}.$$
$Z_1$ is not empty by the claim.
We show that it is both open and closed by picking any $s_1 \in cl(Z_1)$ and showing $s_1 \in int(Z_1)$. 
Since $[x_*,s_1]$ is compact and $O=\{O_z:z\in [x_*,s_1]\}$ is an open covering, there exists  $\{O_1,\dots,O_n\}\subset O$ with $x_*   \in O_1$, $s_1 \in O_n$ and $O_m\bigcap O_{m'}= \emptyset$ for all $m'\geq m+2$. 
On each $O_m$, there exists $\alpha_m,\beta_m$ so that the utility indexes agree by the claim. Also, $O_m$ and $O_{m+1}$ have non-empty intersections with  more than two points, so $(\alpha_{m+1},\beta_{m+1} )=(\alpha_m,\beta_m)$. In particular, $O_1$ intersects $O_{x_*}$ so $\alpha_m=\alpha^*$ for all $m$.
Then  $O_n \bigcap Z \subset Z_1$, i.e. $s_1 \in int(Z_1)$, so  $cl(Z_1)\subset int(Z_1) \subset Z_1 \subset cl(Z_1)$, i.e. $Z_1$ is both closed and open relative to $Z$. Conclude $Z_1=Z$ since $Z$ connected.

Since 	$U^j_1(x)=\alpha U^k_1(x)+\beta$ for all $x \in (x_*,x^*]$ for any interval in the domain, it holds for the whole domain as well. Extend to other categories that intersect $E^i_1 \bigcup E^j_1$ inductively. If there is no intersecting category, we can start again and obtain a (disjoint) interval, the values of $U^i_1$ (and $U^j_1$) on which have no bearing on the DM's choices.
Similar arguments obtain for the other dimensions. Moreover, there is no loss in setting each $\beta=0$. This completes the proof.
\hfill \qedsymbol

\subsection{Proof of Theorem \ref{result: Affine CTM}}
\label{pf: result: Affine CTM}
To save notation, until after Lemma  \ref{lem: utility on IS}, we fix $r$ and write $K^k$ instead of $K^k(r)$ and $\succsim$ instead of $\succsim_r$.
We also identify $x \alpha^k y$ with the alternative $\alpha x \oplus^k (1-\alpha) y$.
Let $\left(U^1,\dots,U^n\right)$ be the additive functions that represent $\succsim_1,\dots,\succsim_n$.
Observe that $U^k(x \alpha^k y)=\alpha U^k(x)+(1-\alpha)U^k(y)$ for any $\alpha$, provided that $x,y,x \alpha^k y \in E^k$.

Recall from Definition \ref{def: IS} that an indifference sequence is a finite sequence of categories with indifference between each succeeding members.
\begin{defn}
	The function $v$ is a \emph{utility for the indifference sequence} $(Q_1,\dots,Q_m)$ if $v$ is an increasing additive utility function on each $Q_k$ and for all $k$, $x,y \in Q_{k} \bigcup Q_{k+1}$:  $x \succsim y \iff v(x) \geq v(y)$.
\end{defn}
\begin{lemma}
	\label{lem: renormalization}
	If $ x^k \in K^k$, $x^l \in K^l$, and $x^k \sim x^l$, then there is $a>0,b \in \mathbb{R}$ such that for $x \in K^k$ and $y \in K^l$, $x \succsim y \iff U^k(x) \geq \alpha U^l(y) + \beta$.
\end{lemma}

\begin{proof}
	W.L.O.G., take $U^k(x^k)=0$.
	There is $\epsilon_k>0$ such that $B_{2\epsilon_k}(x^k) \subset K^k$.
	By CM and Category Continuity, there is $\epsilon_l>0$ such that $ B_{\epsilon_l}(x^l) \subset K^l$ and for all $y \in B_{\epsilon_l}(x^l)$, $x^*=x^k+ \epsilon_k \succ y \succ x^k- \epsilon_k=x_*$. 
	For any $y \in K^l$ and $\alpha$ such that $y\alpha^l x^l \in B_{\epsilon_l}(x^l)$,
	there exists $\beta \in (0,1)$ such that $x^* \beta^k x_* \sim y \alpha^l x^l$ by  Category Continuity, CM, and that $\succsim$ is a weak order.
	Let $V^l(y)= \alpha^{-1} U^k(x^* \beta^k x_*)$. 
	This is well defined, additive, increasing, and ranks alternatives in the same way as $U^l$.
	Thus, $V^l(y)= a U^l(y)+ b$ for some $a>0$ and $b\in \mathbb{R}$.
	
	For any $x \in K^k $ and $y \in K^l $, pick $\alpha\in [0,1]$ such that $x\alpha^k x^k \in B_{\epsilon_k}(x^k)$ and $y \alpha^l x^l \in B_{\epsilon_l}(x^l)$.
	By construction, $y \alpha^l x^l \sim y^\prime$ when $y^\prime \in B_{\epsilon_k}(x^k)$ and $U^k(y^\prime)= \alpha V_l(y)$.
	Thus, $x \alpha^k x^k \succsim y^\prime \sim y \alpha^l x^l$ holds if and only if  $U^k(x) \geq V_l(y)$ and $x \succsim y \iff x \alpha^k x^k \succsim  y \alpha^l x^l$ by AAC since $x^k \sim x^l$, completing the proof.
\end{proof}

For an indifference sequence $(Q_1,\dots,Q_m)$ with utility $v$, we label the range of utilities as $cl(v(Q_k))=[l_k,u_k]$ where $l_k \leq u_k$. Note that we allow $Q_k =Q_l$  for $k \neq l$.

\begin{lemma}For an indifference sequence $(Q_{1},\dots,Q_{m})$, there is an affine, increasing utility $v$ for it.
	
\end{lemma}
\begin{proof}
	The proof is by induction.
	We claim that there is a utility $v^k:X\rightarrow \mathbb{R}$ that is  a utility for the IS $(Q_1,\dots, Q_k)$ for any $k$.
	When $k=1$ or $k=2$, this is true by the above lemmas.
	The induction hypothesis (IH) is that the claim is true for $k=N$.  Consider $k=N+1$. Let $v^N$ be the utility for $(Q_1,\dots, Q_N)$ be index that exists by the IH.
	If $Q_{N+1} \subseteq \bigcup_{i=1}^N Q_i$, then we are done. 
	If not, then for $Q_N=K^l$, there is no loss in normalizing $v^N$ so that it equals $U^l$ on $K^l(r)$. Suppose $Q_{N+1}=K^j(r)$, and let $\alpha,\beta$ be the scalars claimed to exist by Lemma \ref{lem: renormalization}, so that $U^j(x) \geq \alpha U^l(y) + \beta \iff x \succsim_r y$  for $x \in K^k(r)$ and $y \in K^l(r)$. Restricted to $Q_N$,  $v^N=U^l$, so we can define $v^{N+1}(x)=\alpha v^N(x)+\beta$ if $x \in \bigcup_{i=1}^N Q_i$ and $$v^{N+1}(x)=U^j(x)$$ if $x \in Q_{N+1}$.
	Then, if $l<N$ and $x,y \in Q_l \bigcup Q_{l+1}$, then we are done by the IH, since $v^{N+1}(x)\geq v^{N+1}(y) \iff v^N(x)\geq v^N(y)$.
	If $x,y \in Q_N \bigcup Q_{N+1}$, then Lemma \ref{lem: renormalization} and construction implies the result.
	The claim then holds by induction.
\end{proof}

\begin{lemma}
	\label{lem: shorter sequence}
	Fix an indifference sequence $(Q_{1},\dots,Q_{n})$ with utility $v$. 
	If $x^k \in Q_k$ for $k=i,i+1,i+2$ with $x^i \sim x^{i+1} \sim x^{i+2}$, then  $(Q_1,\dots,Q_i,Q_{i+2},\dots,Q_n)$
	is an indifference sequence (after relabeling) with utility $v$.
\end{lemma}

\begin{proof}
	The Lemma is vacuously true for any $1$ or $2$-element IS.
	Fix an IS $(Q_1,\dots,Q_n)$ with $n \geq 3$ and $v$ as above, and suppose $x^k \in Q_k$ for $k=i,i+1,i+2$ with $x^i \sim x^{i+1} \sim x^{i+2}$.
	By transitivity $x^i \sim x^{i+2}$, so $(Q_1,\dots,Q_i,Q_{i+2},\dots,Q_n)$ is an IS; it remains to be shown that $v$ is a utility for it.
	There is an $\epsilon>0$ s.t. $B=B_{\epsilon}(v(x^i))  \subset (l_k,u_k)$ for  $k=i,i+1,i+2$.
	Let $v^{-1}(u):B\rightarrow Q_{i+1}$ be an arbitrary point in $Q_{i+1}$ such that $v[v^{-1}(u)]=u$.
	Now, fix $x \in Q_i$ and $y\in Q_{i+2}$.
	For $\alpha$ small enough, $v(x\alpha^i x^i), v(y\alpha^{i+2} x^{i+2}) \in B$.
	Then $x\alpha^i x^i \sim v^{-1}(v(x\alpha^i x^i))$ and $y\alpha^{i+2} x^{i+2} \sim v^{-1}(v(y\alpha^{i+2} x^{i+2}))$.
	So\begin{eqnarray*}
		x  \succsim  y & \iff & x\alpha^i x^{i}  \succsim  y \alpha^{i+2} x^{i+2}\\
		& \iff & v^{-1}(v(x\alpha^i x^i))  \succsim  v^{-1}(v(y\alpha^{i+2} x^{i+2})) \\
		& \iff & v[v^{-1}(v(x\alpha^i x^i))]  \geq  v[v^{-1}(v(y\alpha^{i+2} x^{i+2}))]\\
		& \iff & \alpha v(x)+(1-\alpha)v(x^i)  \geq  \alpha v(y) +(1-\alpha) v(x^{i+2})\\
		& \iff & v(x) \geq v(y)
	\end{eqnarray*} 
	This establishes the Lemma.
\end{proof}

\begin{lemma}
	\label{lem: triple indifference}
	Fix an indifference sequence $(Q_1,\dots,Q_n)$ with utility $v$. 
	If $(l_1,u_1)\bigcap (l_n,u_n) \neq \emptyset$, then there exists $i$ and  $x^k \in Q_k$ for $k=i,i+1,i+2$ with $x^i \sim x^{i+1} \sim x^{i+2}$.
\end{lemma}
\begin{proof}
	If there is $i$ with $(l_i,u_i) \bigcap (l_{i+2},u_{i+2}) \neq \emptyset$, then there is $u \in \bigcap_{j=i,i+1,i+2}(l_j,u_j)$ so there exists $x_j \in Q_j$ with $v(x_j)=u$ for $j=i,i+1,i+2$ and thus by the hypothesis, $x_i \sim x_{i+1} \sim x_{i+2}$.
	We show there exists such an $i$ by contradiction.
	If $l_{i+2}>u_i$ for all $i$ or $l_i>u_{i+2}$ for all $i$, then $(l_1,u_1) \bigcap (l_n, u_n) =\emptyset$, a contradiction.
	So there must exist $i$ such that [$l_{i+2}>u_i$ and $l_{i+2}>u_{i+4}$] or [$u_{i+2}<l_i$ and $u_{i+2}<l_{i+4}$].
	In the first case, $l_{i+2} \in (l_{i+1},u_{i+1}) \bigcap(l_{i+3},u_{i+3})$; in the second, $u_{i+2} \in (l_{i+1},u_{i+1}) \bigcap(l_{i+3},u_{i+3})$.
	In either case, we have a contradiction.
\end{proof}

\begin{lemma}
	\label{lem: utility on IS}
	Fix an indifference sequence $(Q_1,\dots,Q_n)$ with utility $v$.  
	Then for all $x,y \in \bigcup_i Q_i$, $x \succsim y \iff v(x) \geq v(y)$.
\end{lemma}
\begin{proof}
	This is clearly true if $n=1$.
	(IH) Suppose the claim is true for any IS with $m<n$ elements.
	Fix an IS $(Q_1,\dots,Q_n)$ with utility $v$.
	If $x \notin Q_1 \bigcup Q_n$ or $y \notin Q_1 \bigcup Q_n$, then the claim immediately follows from the IH, and clearly holds if $x,y\in Q_i$ for some $i$.  So it suffices to consider arbitrary $x \in Q_1$ and $y \in Q_n$.
	By Lemmas \ref{lem: shorter sequence} and \ref{lem: triple indifference}, if $(u_1,l_1) \bigcap (l_n,u_n) \neq \emptyset$, we can form a shorter IS from $Q_1$ to $Q_n$ and the claim then follows from the IH.
	
	There are two cases to consider:  $l_n>u_1$ and $u_n<l_1$.
	Consider $l_n>u_1$.
	The range of $v$ restricted to $\bigcup_{i=1}^{n-1}Q_i$ is dense in $\bigcup_{i=1}^{n-1}(l_i,u_i)=(\bar{l},\bar{u})$.
	Note $l_n \in (\bar{l},\bar{u})$ since $x_{n-1} \sim y_{n}$, so $(l_{n-1},u_{n-1})\bigcap (l_n,u_n) \neq \emptyset$.
	Then $(l_n,v(y))$ is an open interval having a non-empty intersection with $(\bar{l},\bar{u})$.
	Since the range of $v$ is dense in $(\bar{l},\bar{u})$, there exists $y' \in  Q_{n'} $ with $l_n<v(y') <v(y)$.
	Since $l_n>u_1$, $n'>1$.
	Then $(Q_1,\dots,Q_{n'})$  and $(Q_{n'},\dots,Q_n)$ are both ISes with strictly less than $n$ elements. Applying the IH, $y^\prime \succ x$ and $y\succ y^\prime$.
	Conclude using transitivity that $y\succ x$.
	Similar arguments obtain the desired conclusion when $u_n<l_1$.
\end{proof}

Define $\bowtie_r$ as in the proof of Theorem \ref{thm: weak CTM}, and let $A_1,\dots,A_n$ be the distinct indifference classes of $\bowtie_r$.
Again using Lemma \ref{lem: ordering of ISes}, we can relabel so that $x \in A_i$ and $y \in A_{i+1}$ implies $x\succ_r y$.
By Lemma \ref{lem: utility on IS}, there is $v_i$ on $A_i$ so that $v_i$ is additive and increasing within categories and $x\succsim y \iff v_i(x)\geq v_i(y)$ for all $x,y\in A_i$.

By Unbounded and Lemma \ref{lem: ordering of ISes}, every positive unbounded region (if any) is a subset of $A_1$, and every negative unbounded region (if any) is a subset of $A_n$. If one region is both positive and negative unbounded, then $n=1$. Therefore, $v_i(A_i)$ is bounded for all $i\in (1,n)$, and $v_n(A_n)$ is bounded above whenever $n>1$.
Define $V(x)=v_1(x)$ for all $x \in A_1$. For $x \in A_i$ with $i>1$, define $V(x)$ recursively by \[
V(x) = v_i(x) - \sup_{y \in A_{i}} v_i(y) + \inf_{y\in A_{i-1}} V(y) -1.
\] 
Observe $V(\cdot)$ is a positive affine transformation of $v_i(\cdot)$ when restricted to $A_i$, and if $x \in A_i$, $y\in A_j$ and $i>j$, then $V(x)>V(y)$. Thus $V$ represents $\succsim_r$ and,  when restricted to any given region, is affine and increasing.

Defining $U^k(\cdot|r)$ as the (unique) affine transformation of $U^k$ so it agrees with $V$ on $K^k(r)$ establishes that $\succsim_r$ is an Affine CTM.
Since $r$ was arbitrary, this establishes that each $\succsim_r$ has such a  representation. Conclude that $\{\succsim_r\}$ conforms to Affine CTM, completing the proof. \hfill \qedsymbol

\subsection{Proof of Theorem \ref{result: strong CTM}} 
\label{pf: result: strong CTM}

 Without loss of generality, normalize so that $U^1(\cdot|r)=U^1(\cdot|r^\prime)$ for all $r,r^\prime$. 
Suppose $U^k(\cdot|r) \neq U^k(\cdot|r^\prime)$ for some $r,r^\prime$ and some $k$. Then, let $\bar{\epsilon}=d(r,r^\prime)$ and pick a sequence $\hat{r}_n \rightarrow \hat{r}$ such that: $U^k(\cdot|\hat{r}_n)\neq U^k(\cdot|r)$, $\hat{r}_n \in B_{\bar{\epsilon}}(r)$ for all $n$, and  $d(\hat{r}_n,r) \rightarrow \inf \{d(r^\prime,r):U^k(\cdot|r) \neq U^k(\cdot|r^\prime)\}$. Since $\hat{r}_n \in cl(B_{\bar{\epsilon}}(r))$, there is no loss in assuming this sequence converges.
Similarly, let $r_n$ be a sequence in $B_{\bar{\epsilon}}(r)$  such that $r_n \rightarrow \hat{r}$ and $U^k(\cdot|r) = U^k(\cdot|r_n)$. 

By hypothesis and that each $K^k(r)$ is open, there exists $\epsilon>0$, $x^k$ and $x^1$ such that  $B_{2\epsilon}(x^k) \subset K^k(\hat{r})$, $B_{2\epsilon}(x^1) \subset K^1(\hat{r})$, and $x^k \sim_{\hat{r}} x^1$. By continuity of the region functions, $B_{\epsilon}(x^k) \subseteq K^i(\hat{r}_n)\cap K^i(r_n)$ and $B_{\epsilon}(x^1) \subseteq K^1(\hat{r}_n)\cap K^1(r_n)$ for $n$ large enough. For $z$ close enough to $x^k$, there exists $y(z)\in B_{\epsilon}(x^1)$ such that $z \sim_{\hat{r}} y(z)$.	 But then by SC, $z \sim_{r_n} y(z)$ and $z \sim_{\hat{r}_n} y(z)$. 	Thus $U^k(z|r_n)=U^1(y(z)|r_n)=U^1(y(z)|\hat{r}_n)=U^k(z|\hat{r}_n)$ for all $z$ close enough to $x_k$, implying that $U^k(\cdot|r_n)=U^k(\cdot|\hat{r}_n)$, a contradiction. Conclude $U^k(\cdot|r)=U^k(\cdot|r^\prime)$ for all $r,r^\prime$.
\hfill\qedsymbol
\subsection{Examples from Table \ref{table:comparison}}
\label{sec: examples table}
Example \ref{ex: BGS cancellation} shows that BGS violates Cancellation and inspecting Figure \ref{fig:models} shows it violates Monotonicity. It remains to show that TK violates Reference Irrelevance and that MO violates Cancellation. This is established by the following two examples.
\begin{example}[TK violates Reference Irrelevance]
	Consider a TK model with $\lambda_1=\lambda_2=2$.
	Then, for $r=(10,10)$, $x=(12,12)$ and $y=(9,16)$, $y \succsim_r x$ since $(12-10)+(12-10)=2(9-10)+(16-10)$.
	For $r'=(11,11)$, $x \succ_r y$ since $(12-11)+(12-11)>2(9-11)+(16-11)$. But $x\in R^{GL}_1(r)\bigcap R^{GL}_1(r')$ and $r\in R^{GL}_2(r)\bigcap R^{GL}_2(r')$, so the family violates Reference Irrelevance.
\end{example}

\begin{example}[MO violates Cancellation]
	Let $Q(r)=\left\{ x \in X: x_1/2+x_2 >r_1/2+r_2 \right\}$ and  $c(r)=1$. Then, let $x=(2,1)$, $y=(1,2)$, $z=(4,4)$, and $r=(0.9,1.9)$. Since $(x_1,z_2)=(2,4)\succsim_r (4,2)=(z_1,y_2)$ and $(z_1,x_2)=(4,1)\succsim_r (1,4)=(y_1,z_2)$ because all four points belong to $Q(r)$, cancellation requires that $x\succsim_r y$. However, $x\notin Q(r)$, so $y \succ_r x$, so cancellation does not hold.
\end{example}

\subsection{Other CTM}\label{appendix:otherCTM}

\subsubsection{Quasi-Hyperbolic Model}  

Let the pair $(c, t)$ represent consumption of $c$ at time  $t$.  Formally, we define categories according to $K^{QH}=(K^{short},K^{long})$ where $K^{short}(r_c,r_t)=\{(c,t)| t < r_t\}$ and $K^{long}(r)=\{(c,t)| t > r_t\}$. 
The utility function is 
$$V_{QH}(c,t|r)=\left\{ 
\begin{array}{lc}
(\beta\delta)^t u(c)  & \text{ if } (c,t)  \in K^{short}(r)  \\
\beta^{r_t} \delta^t u(c)  & \text{ if } (c,t)  \in K^{long}(r)  \\
\end{array}	
\right.$$
where $0  <  \delta < 1 $ and $0<\beta\leq 1$.
The model is additively separable after taking logs,   
so it is a special case of CTM. 
It exhibits present bias when $\beta<1$: there exist values $c>c'>0$ so that the DM prefers $(c,\tau ) \succsim_{r} (c',\tau+1)$ if and only if $\tau<r_t-1$.%
\footnote{For instance $u(c)=1$ and $u(c')=(\beta \delta)^{-1}$.}
Figure \ref{fig:Time} plots its indifference curves.

\begin{figure}[h]
	\begin{center}
		\includegraphics[width=0.3\textwidth]{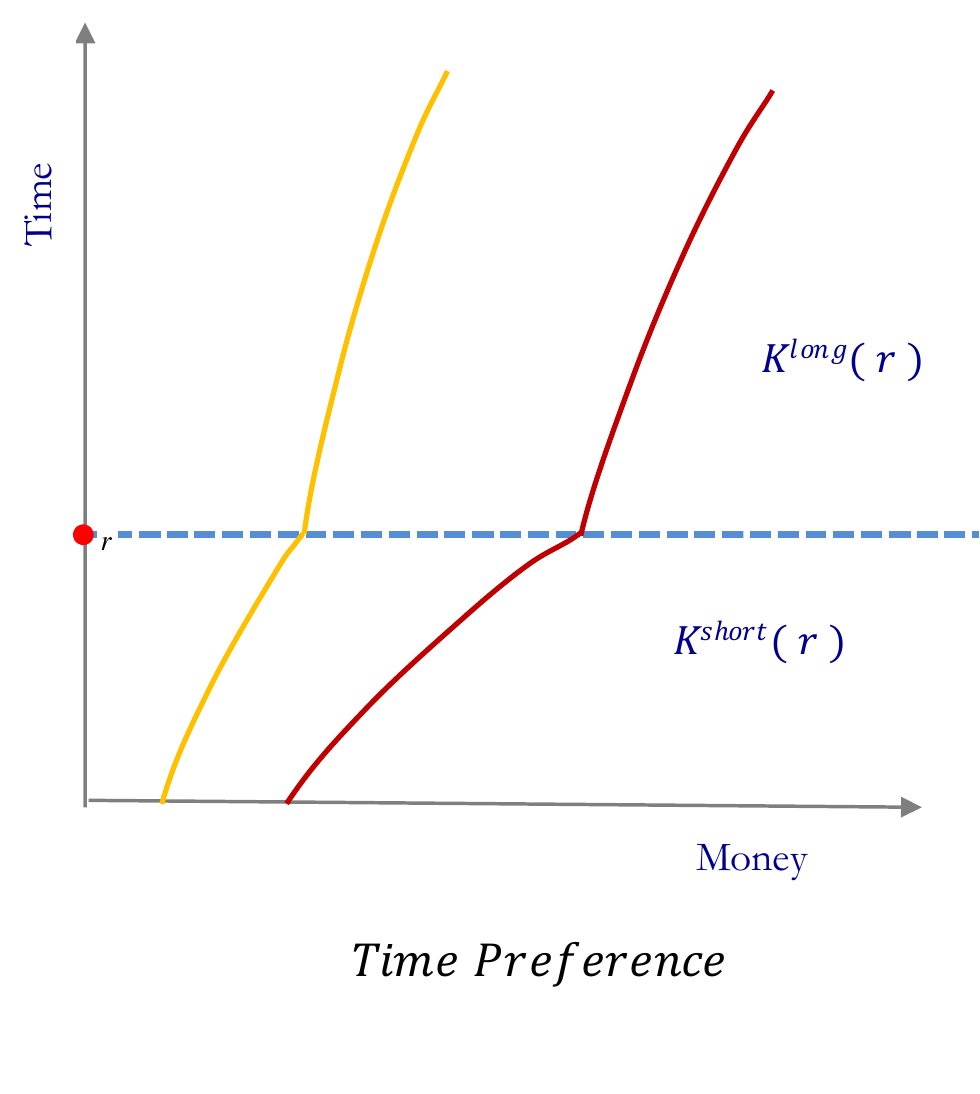}
	\end{center}
	\caption{CTM for Dated Rewards}
	\label{fig:Time}
\end{figure}

\subsubsection{Inequality Aversion Model} The category function is $\mathcal{K}^{RIA}= (K^E,K^G)$ where   $$K^G(r)=\{x \in X: {x_1}-{r_1} > {x_2}-{r_2} \}   \text{ and } K^E(r)=\{x \in X: {x_1}-{r_1} < {x_2}-{r_2} \} $$
The set $K^G(r)$ contains all allocations where individual $1$ is advantaged relative to the individual $2$, and $K^E(r)$ all those where she is disadvantaged.%
\footnote{In general there are $2^{n-1}$ categories, corresponding to envy or guilt for each binary comparison with every other individual. For instance, there are 4 categories with $n=3$: $EE$, $EG$, $GE$, and $GG$.}

The DM  feels guilty if her own relative gain  is higher than the other's relative gain. Otherwise, the  DM is envious of the other. 
Hence, a social allocation  $x$ is evaluated according to
\begin{equation*}
\label{eq:FS rep2}
V_{RIA}(x|r)=\left\{ 
\begin{array}{lc}
x_1 -\alpha [(x_1-r_1) - (x_2-r_2)]  & \text{ if } x \in K^E(r)  \\
x_1 -\beta  [(x_2-r_2) -(x_1-r_1)]  & \text{ if } x \in K^G(r)  
\end{array}
\right.
\end{equation*}
where $\alpha \geq \beta \geq 0$ and $\beta<1$.
Observe that when $r_i=r_j$ for all $i$ and $j$ (the equitable outcome), the utility function reduces to that of \cite{fehrschmidt99}. 
Also, the model is an Affine CTM, and a Strong CTM for the restricted set of reference points with $r_1=r_2$.

\begin{figure}[h]
	\begin{center}
		\includegraphics[width=0.7\textwidth]{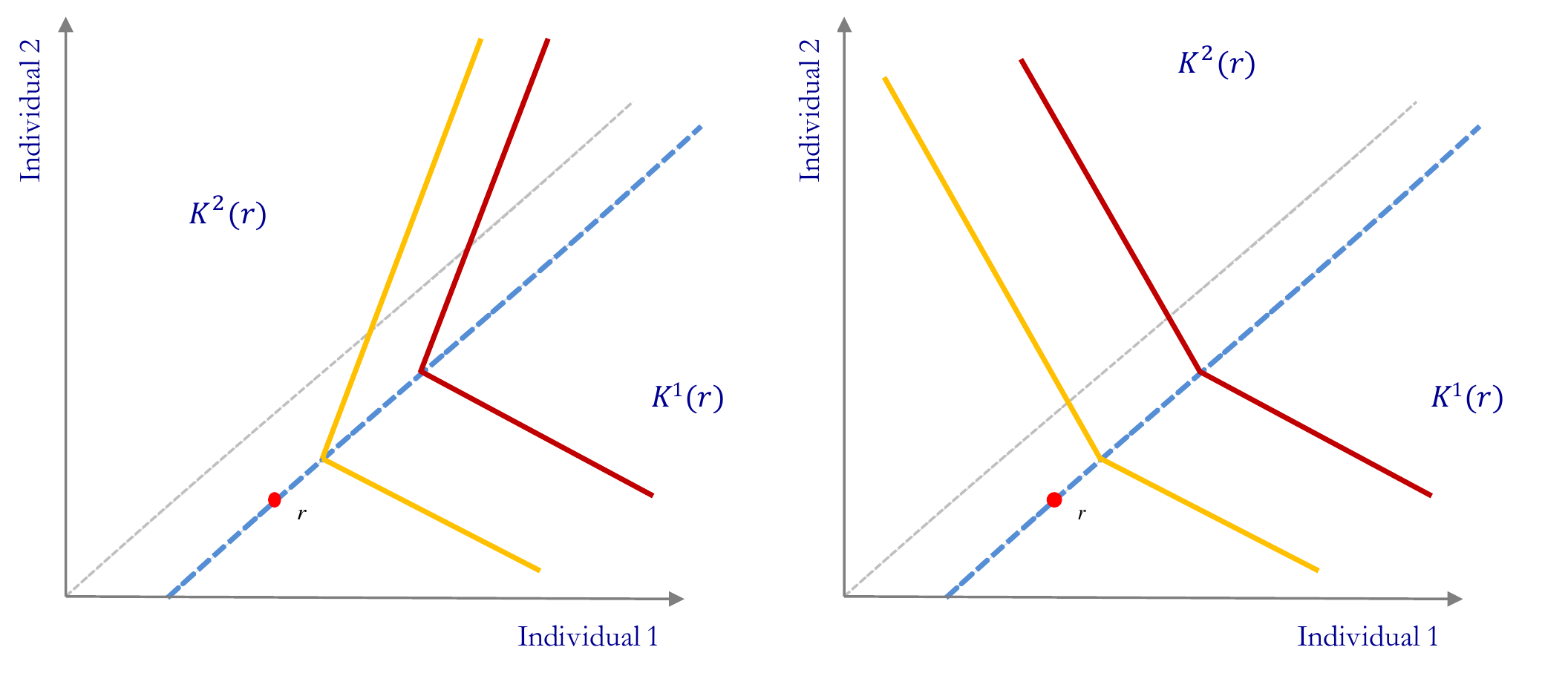}
	\end{center}
	\caption{Left: Relative Inequality Aversion and Right: Reference-Dependent Distributional Preferences }
	\label{fig:FS}
\end{figure}

\subsubsection{Distributional Preferences:}
\cite{charness2002understanding} argue that people care about both their own utility and social welfare as a whole.
They maximize a weighted average of the two. Social welfare is a weighted average of society's total utility and the utility of worst-off individual.
We propose a natural extension of their model with an exogenously given  reference point. We call this model \emph{Reference-Dependent Distributional Preferences (RDDP)}.
Formally, categories are given by $\mathcal{K}^{CR}= (K^1, K^2)$ where  $$K^j(r)= \left\{x \in X: j  = \arg\min_i  (x_i-r_i) \right\}.$$ 
Each category corresponds to the individual with the worst relative payoff.
If $j=1$, then the DM is behind and wants to catch up.
If $j=2$, then the DM is ahead and is more willing to help the other to catch up.%
\footnote{With $n$ individuals, there are $n$ categories, each corresponding to the identity of the worst-treated individual in terms of relative consumption, $(x_i-r_i)$.  The utility function reduce to $V_{CR}(x|r)=(1- \lambda)(x_1-r_1)+ \lambda [\delta  \min\{x_1-r_1,\dots,x_n-r_n\}+(1-\delta) \sum_k (x_k-r_k)] $.
While RDDP and RIA have category functions that coincide with $n=2$ individuals, their category functions diverge for all other $n$. Figure \ref{fig:FS} reveals that the behavior necessarily differs even with $n=2$.
}

In RDDP, the DM puts extra weight on the consumption of the individual who is furthest behind. Formally, she evaluates a social allocation $x$ with reference $r$ according to
\[
V_{CR}(x|r)
=\left\{	\begin{array}{cc}
(1- \lambda  )(x_1-r_1)+ \lambda[\delta  (x_1-r_1)+(1-\delta)\sum_k(x_k-r_k)]  &\text{ if } x \in K^1(r) \\
(1- \lambda  )(x_1-r_1)+ \lambda[\delta  (x_2-r_2)+(1-\delta)\sum_k(x_k-r_k)] &\text{ if } x \in K^2(r) 
\end{array}\right.
\] where  $\delta, \lambda \in (0,1)$.
Utility is increasing in the DM's own consumption, the minimum of all individuals' payoffs, and the total of all individuals' payoffs. Hence, the DM is willing to give up more of her own consumption to increase that of the worst-off individual than that of one of others.
The parameter $\delta$ measures the degree of concern for helping the worst-off individual (Rawlsian) versus maximizing the total social payoffs (Utilitarian), and
 $\lambda  $ measures how the DM balances social welfare with her own material payoff.
Note that if $r_i=r_j$ for all $i$ and $j$, the utility function is cardinally equivalent to that of \cite{charness2002understanding}.%
\footnote{The authors assume $U(x)=(1- \lambda)x_1+ \lambda [\delta  \min\{x_1,\dots,x_n\}+(1-\delta) \sum_k x_k] $ (see their Appendix 1). Pick any allocation $x$ and reference point $r$ so that $r_j=r_k=r^*$ for every $j$ and $k$. Let $i^*\in \arg\min_j x_j$.
Subtracting the same constant from each element in a set does not change the minimizer, so $V_{CR}(x|r)=(1-\lambda)(x_1-r_1)+\lambda \delta (x_{i^*}-r_{i^*})+\lambda (1-\delta)\sum_k (x_k-r_k)=U(x)-\lambda \delta r_{i^*}-(1-\lambda) r_1-\lambda \sum_k r_k$.
Since $r_{j}=r^*$ for all $j$, $V_{CR}(x|r)=U(x)-(\lambda n +(1-\lambda(1-\delta)))r^*$, i.e. it is an affine transformation of $U(x)$ and this tranformation does not depend on $x$.}
The model is an Affine CTM.

\subsection{Other models and CTM}
\label{sec: behavioral models that are not RPM}
In this subsection, we present the functional forms of the other models we discussed, and show that they are not CTM.

\begin{itemize}
%

	\item \cite{gabaix2014sparsity} assumes a rational DM would  maximize $u(a,w)$ but actually maximizes \[u \left(a,(w_1 m^*_1,\dots,w_n m^*_n) \right)\] where \[
	m^* \in \arg \min_{m \in [0,1]^n} \frac12 \sum_{i,j}(1-m_i)\Lambda_{ij} (1-m_j)+\kappa \sum_{i} m_i^\alpha 
	\]
	where $\Lambda_{ij}$ incorporates the ``variance'' in the marginal utility of dimensions $i$ and $j$. When $n$ is large, $m^*_i$ is often zero, so $(w_1 m^*_1,\dots,w_n m^*_n)$ is a ``sparse'' vector.
	\item \cite{TK91} refer in general to $$V_{CTK}(x|r)=\sum_i v_i(u_i(x_i)-u_i(r_i)) $$  where $v_i$ is concave above $0$ and convex below
	\item \cite{BGS_MAC} and the continuous form of the salient thinking model has $$V_{CBGS}(x|r) =w(x_1,r_1) x_1 + w(x_2,r_2) x_2$$ where $w$ has the same properties as a salience function. 
	\item \cite{munro2003theory} use the functional from $$
	V_{MS}(x|r)=A(r)\left( \sum_i \gamma_i r_i^{\rho-\beta} x_i^\beta \right)^{\frac{1}{\beta}}
	$$
	\item \cite{bhatia2013attention} assume that the DM chooses the bundle $x$ that maximizes \[
	U(x|r)= \alpha_1(r_1) [V(x_1)-V(r_1)]+\alpha_2(r_2) [V(x_2)-V(r_2)]
	\]
	given that a reference point $r$, where each $\alpha_i$ is increasing and positive.

\end{itemize}	

The first fails to be CTM, as the indifference curves have the same slope everywhere for a fixed context.  If they were CTM, then they would necessarily have only a single region. Single region CTM coincides with the neoclassical model.
The final four explicitly take into account a reference point.
In all four, it is easy to see that the reference point affects the marginal rate of substitution between attributes.
This implies a violation of weak reference irrelevance for any given category function: any two points in the same category that are indifferent to each other necessarily remain so for a sufficiently small change in the reference point.

\subsubsection{Non-increasing CTM}
For simplicity, we have so far focused on increasing CTM.
This is a desirable feature in consumer choice, but models of social preference often violate this property.
For instance, inequality-averse individual $1$ prefers to increase the allocation to individual $2$ from $x$ to $y$ when  she feels guilty but not when she is envious.
However, she always prefers increasing the allocation to $2$ in an allocation categorized as guilty, and to decrease in any categorized as envious.
This contradicts Category Montonicity, suggesting the following weakening.

\begin{ax1}
	[Consistent Preference within Category, CPC]
	For each category $k$, there exists a set of attributes $P^k$ so that if $x_j \geq y_j$ for all $j \in P^k$, $y_i \geq x_i$ for all $i\notin P^k$, and $x \neq y$, then $y \not \succsim^{k*} x$.
\end{ax1}
The set $P^k$ contains the attributes for which an increase positively affects the DM's evaluation. CPC requires that the set of positive attributes in a category does not depend on the reference point. For the two-person-RIA model, the set  for the ``guilty'' category is $\{1,2\}$ since she strictly prefers increasing everyone's allocation, but the set for the ``envious'' one is $\{1\}$ -- she prefers more for herself but dislikes others having even more.
Note that CM is the special case of CPC where $P^k$ includes every dimension for every category.

A  CTM is characterized by all the properties of an increasing CTM, except where CM is replaced by CPC. The proof is a straightforward generalization of earlier one, so it is omitted.

\subsection{Proof of Proposition \ref{prop: general category ID}}
 Suppose that $\{\succsim_r\}_{r\in X}$ has a CTM and fix a category $k$ with $LIS^k(x)\neq LIS^l(x)$ for every $x \in X$ and category $l \neq k$.
Consider a category $k$ and reference $r$.
Define \[K= \{x \in X: \exists \epsilon>0\ s.t.\ \forall y \in B_{\epsilon}(x),\ y \sim_r x \iff U^k(x) = U^k(y)\}. \]
We show $int(K)= K^k(r)$.
Let $x \in K^k(r)$. Then, there exists a neighborhood $O \ni x$ with $O \subset K^k(r)$ since $K^k(r)$ open.
By the representation, for any $y\in O$, $x \sim_r y$ if and only if $U^k(y)=U^k(x)$, so picking any $\epsilon>0$ so that $B_\epsilon(x) \subset O$ shows that $ x\in K$. Since $K^k(r)$ is open and $K^k(r) \subset K$, $K^k(r) \subset int (K)$.

To show the reverse inclusion, suppose that $x \in K^l(r)$ for category $l \neq k$.
Since $LIS^k(x) \neq LIS^l(x)$, for any neighborhood $O \ni x$ there exists $y \in O$ so that either $U^k(y) \neq U^k(x)$ and $U^l(x) = U^l(y)$ or $U^k(y) = U^k(x)$ and $U^l(x) \neq U^l(y)$.
In particular this applies to $O'=O \cap K^l(r)$, so either there exists $y \in O'$ so that either $y \sim_r x$ and $U^k(y) \neq U^k(x)$ (in the first case) or $y \not \sim_r x$ and $U^k(y) = U^k(x)$ (in the second). Hence, $x \notin K$. Since $x$ is arbitrary, we have $K^l(r) \cap K = \emptyset$. Since $K^l$ is open, we have $int(K) \cap cl(K^l(r)) = \emptyset$.
Since $l$ was arbitrary, $cl(\bigcup_{l\neq k} K^l(r))\cap int(K)=\bigcup_{l\neq k} cl(K^l(r)) \cap int(K)= \emptyset$ since there are finitely many categories. Since the categories are dense, $int(K) \subset cl(K^k(r))$, and it follows that $int(K) \subset int( cl(K^k(r)))=K^k(r)$ since $K^k(r)$ is a regular open set.
Conclude $int(K)=K^k(r)$, and that we can identify $K^k(r)$ for any $k$ and $r$.

\subsection{Proofs and extra material from Section \ref{sec:revealing}}
\label{sec: prop: revealing regions}
\begin{proof}[Proof of Proposition \ref{prop: revealing regions}]
 Suppose that $\{\succsim_r\}_{r\in X}$ has a BGS representation. From Proposition \ref{prop: general category ID}, we need to show that $LIS^1(x) \neq LIS^2(x)$ for all $x$. 
Fix any $x$ and take $r^1_x=(x_1/2,x_2)$ and $r^2_x=(x_1,x_2/2)$. By S4, $x \in K^i(r^i_x)$ for $i=1,2$.
Since $K^1(r^1_x)\cap K^2(r^2_x) $ is open and contains $x$, there exists a neighborhood $O_x$ of $x$ contained in it.
For $y \in O_x$, $y \sim_{r^i_x} x$ if and only if $w^i_1/w^i_2 [u_1(y_1)-u_1(x_1)]=u_2(x_2)-u_2(y_2)$.
Since $u_1$ and $u_2$ are strictly increasing and $w^1_1/w^1_2>w^2_1/w^2_2$, $LIS^1(x) \neq LIS^2(x)$.
Hence, Proposition \ref{prop: general category ID} is applicable and the categories are uniquely identified. Moreover,\[
K^i(r)=\hat{K}^i(r)=int \left\{x \in X: \exists \epsilon>0\ s.t.\ \forall y \in B_{\epsilon}(x),\ y \sim_r x \iff y \sim_{r^i_x} x \right\}
\]
using the above arguments and taking $\epsilon$ so that $B_\epsilon(x)\subset O_x$.
\end{proof}

\begin{proposition}
Let $\{\succsim_r\}_{r\in X}$ be  a CTM where each category is connected. For any reference $r$ such that 
 $U^k(x|r) \neq U^l(x|r)$ for every $x\in X$ and categories $l,k$ with $l \neq k$, the category function is uniquely identified for $r$.
\end{proposition} 
\begin{proof}
For any $x^*$, let $O$ be the $\subseteq$-largest connected, open set connected so that for every $x\in O$ there is an $\epsilon>0$ so that $\{y \in B_\epsilon(x): y \succsim_r z\}$ and $\{y \in B_\epsilon(x): z \succsim_r y\}$ are closed for each $z \in B_\epsilon(x)$.
A maximal set with this property exists by Zorn's Lemma, and is unique since any two such maximal sets contain $x^*$, so their union is also a maximal set. 
If $x^* \in K^k(r)$, we claim that $O=K^k(r)$. 

First, note  $O \subset \bigcup_{l=1}^n K^l(r)$, where $n$ is the number of categories. For $x \in X$ so that $x \notin \bigcup_{l=1}^n K^l(r)$, there are categories $i,j$ so that $x\in bd K^i(r)$ and $x\in bd K^j(r)$.
WLOG, $U^i(x|r)>U^j(x|r)$.
For any $\epsilon>0$, there exists $x' \in K^i(r) \cap B_\epsilon(x)$ with $U^i(x|r)>U^i(x'|r)>U^j(x|r)$, $x''\in K^j(r) \cap B_\epsilon(x)$ with $U^i(x|r)>U^j(x''|r)>U^j(x|r)$, and sequences $x'_n \in K^i(r)\cap B_\epsilon(x)$ and $x''_n \in K^j(r)\cap B_\epsilon(x)$ so that $x'_n \rightarrow x$ and $x''_n \rightarrow x$. Since $x' \succ_r x''$, either $x' \succ_r x$ or $x \succ_r x''$.
In the former case, $x'_n \succsim_r x'$ for all $n$ large enough but $x' \succ_r x$; in the latter, $x'' \succsim_r x''_n$ for all $n$ large enough but $x \succ_r x''$.
In either case we obtain a contradiction. 

Now, $K^k(r)$ and $K^{-k}(r)=\bigcup_{l\neq k} K^l(r)$ are disjoint, open sets whose union contains $O$. Hence either $K^k(r) \cap O=\emptyset $ or $K^{-k}(r) \cap O= \emptyset$. The former is impossible, so $K^k(r) \supseteq O$. But clearly $K^k(r)$  is a connected, open set satisfying the condition of $O$, so $K^k(r) \subseteq O$. Conclude $K^k(r) =  O$.
\end{proof}

Consider the following category utility functions. $U^1(x_1,x_2|r)=x_1+x_2$ and \[
U^2(x_1,x_2|r)=\left\{ \begin{array}{cc}
x_1+x_2 & \text{ if }x_1+x_2\leq 1 \\
2(x_1+x_2) -1 &\text{ if }x_1+x_2> 1 \\
\end{array}\right.
\]
For any category function, the boundary between the categories can be identified on the set $\{x:x_1+x_2>1\}$ but not on the other points. Intuitively, the DM evaluates objects in this set the same regardless of their categorization and so we cannot hope to identify their category.
When the category does not affect the DM's choice, the revealed preference approach cannot distinguish the two and non-choice data must be used.

\subsection{Proof of Proposition \ref{thm:BGS preference}}
	$K$ satisfying S0-S4 implies that $E^1=E^2=\mathbb{R}^n_{++}$, so the structure assumption is satisfied.
	Moreover, Theorem \ref{thm: BGS regions} gives that the categories are generated by a salience function.
	The axioms allow us to apply Theorems \ref{thm: RI} and \ref{result: strong CTM} to get a Strong CTM representation of the family with reweighted utility indexes.
	Hence, $$U^k(x)= w^k_1 u_1(x_1)+w^k_2 u_2(x_2)+\beta^k$$ for each $x \in X$. 
	 
	There is no loss in normalizing so that $\beta^1=0$.
	Pick $x \in X$ with $x_1>x_2$, and by S4 observe that $x \in K^1(r)$ for $r=(x_1,x_2/2)$ and $x \in K^2(r')$ for $r'=(x_1/2,x_2)$.
	Since $K^1(r)$ and $K^2(r')$ are open, there exists $\epsilon>0$ so that $B_\epsilon(x) \subset K^1(r) \bigcap K^2(r)$. Since $U^1$ is continuous and increasing, there is $y \in B_\epsilon(x) $ with $y_1<x_1$ so that $U^1(y)=U^1(x)$, i.e. $y \sim_r x$; this $y$ necessarily has $y_2>x_2$ by CM. 
	Then, SDO implies $y \succ_{r'} x$, i.e. $U^2(y)>U^2(x)$, which requires $w_2^1/w_2^2<w^1_1>w^2_1$.
	We can incorporate $\beta^2$ into $u_2$ by replacing it with $u_2+\beta^2/(w^2_2-w^1_2)$ or into $u_1$ by replacing $\beta^2$ into $u_1$ by replacing it with $u_1+\beta^2/(w^2_1-w^1_1)$. 
	At least one does not involve dividing by zero, as otherwise $w^2_i=w^1_i$ for $i=1,2$.
	\hfill\qedsymbol
\subsection{TK} 
\label{pf: thm:TKinRPM}
This subsection states and proves a characterization theorem for TK.
	\begin{proposition}
	\label{thm:TKinRPM}
	A family of preferences $\{\succsim_r\}_{r\in X}$ has a TK representation if and only if it is an Affine CTM with a gain-loss regional function that satisfies Reference Interlocking, Monotonicity, Cancellation, and continuity of each $\succsim_r$.
\end{proposition}
\citet[p. 1053]{TK91} provide an alternative axiomatic characterization of the model, and our result makes heavy use of their theorem. 
\begin{proof}
	Necessity follows from the discussion above and TK's theorem. To show sufficiency, we rely on TK's theorem, which states that any monotone, continuous family of preference relations that satisfies cancellation, sign-dependence and reference interlocking has a TK representation.
	Given our assumptions, we need to show that $\{\succsim_r\}$ satisfies sign-dependence and reference interlocking.
	
	TK say that $\{\succsim_r\}$ satisfies sign-dependence if ``for any $x,y,r,s \in X$, $x\succsim_r y \iff x \succsim_s y$ whenever $x$ and $y$ belong to the same quadrant with respect to $r$ and with respect to $s$, and  $r$ and $s$ belong to the same quadrant with respect to $x$ and with respect to $y$.''
	This happens if and only if $x \in K_k(r) \bigcap K^k(s)$ and $y \in K^k(r) \bigcap K^k(s)$ for some $k \in \{1,2,3,4\}$.
	Then, sign-dependence is exactly an implication of Affine CTM, since $U^k(\cdot|r)=\alpha U^k(\cdot|s)+\beta$ for $\alpha>0$.
	
	TK say that $\{\succsim_r\}$ satisfies reference interlocking if ``for any $w,w',x,x',y,y',z,z'$ that belong to the same quadrant with respect to $r$ as well as with respect to $s$, $w_1=w'_1$, $x_1=x'_1,y_1=y'_1,z_1=z'_1$ and $x_2=z_2,w_2=y_2,x'_2=z'_2,w'_2=y'_2$, if $w\sim_r x$, $y \sim_r z$, and $w' \sim_s x'$ then $y' \sim_s z'$.''
	The assumptions on quadrants imply that $w,w',x,x',y,y',z,z' \in K^k(r) \bigcap K^l(s)$ for some $k,l \in \{1,2,3,4\}$. 
	Since $y',z' \in K^l(s)$, the conclusion follows immediately from RI.
\end{proof}

\subsection{Example \ref{ex: regions BGS}}
\begin{example}
	\label{ex: regions BGS}  
	 The  categories plotted in Figure \ref{fig:categories} are described formally below. They all satisfy S0-S3, but only a subset of the other properties. 
	\begin{enumerate}
		\item The category function\[
		K^1(r)=\{x:s^1(x_1,r_1)>s^1(x_2,r_2)\}\text{ and }K^2(r)=\{x:s^1(x_1,r_1)<s^1(x_2,r_2)\}
		\] 
		where $s^1(x,r)=\frac{\max \{x,r\}^2}{\min \{x,r\}}$  violates S4-S6. Note $s^1$ is not a salience function since it is not grounded: $s(a,a)=a$ for $a >0$. Then $(a,b+\epsilon),(a,b) \in K^1(a,b)$ for all $a > b$ and small enough $\epsilon>0$, contradicting S4 and S6, respectively. Also note $s^1(a,a)=s^1(\sqrt{a},1)$ for $a >0$. Hence, $(a,\sqrt{a}) \notin K^1(a,1)$ but $(a+\epsilon,\sqrt{a}) \in K^1(a+\epsilon,1)$ for every $\epsilon>0$, violating S5.
		\item The salience function  $s^2(x,r)= |x^2-r^2|$ generates regions that satisfy S0-S4 but violate S5 and S6. 
		Observe that $(2,\sqrt{5}) \notin K^1(1,\sqrt{2})$ since $s^2(2,1)=\sigma(\sqrt{5},\sqrt{2})=3$, but $(2+\epsilon,\sqrt{5}) \in K^1(1+\epsilon,\sqrt{2})$ for any $\epsilon>0$ since $s^2(2+\epsilon,1+\epsilon)=3+2\epsilon > 3$, contradicting S5. It is routine to verify S4 by differentiating. Also, $x=(2,2)$   and $r=(4,1)$ have $x_1 x_2=r_1 r_2$, but $s^2(2,4) > s^2(2,1)=$, so $x \in K^1(r)$, contradicting S6.
		\item The salience  function $s^3(x,r)= |\sqrt{x}-\sqrt{r}| $ generates regions that satisfy S0-S5 but violate S6. Also, $x=(2,2)$   and $r=(4,1)$ have $x_1 x_2=r_1 r_2$, but $s^3(2,4)> s^3(2,1)$, so $x \in K^1(r)$, contradicting S6.  Differentiating establishes S4 and S5.
		\item The  salience function  $s^4(x,r)= \frac{ \max \{x,r\}}{\min\{x,r\}}$ generates regions that satisfy S0-S6.
	\end{enumerate}
\end{example}

\subsection{Proof of Theorem \ref{thm: BGS regions}}
\label{pf: thm: BGS regions}
We first prove the following lemma. 
\begin{lemma}
	If $\mathcal{K}$ is a category function, then for any $\epsilon>0$ and $x$ so that $B_\epsilon(x) \subset K^i(r)$, there exists $\delta>0$ so that $B_{\epsilon/2}(x) \subset K^i(r')$ for all $r' \in B_\delta(r)$.
	\label{lem: usc}
\end{lemma}
\begin{proof} 
	Let $\mathcal{K}$ is a category function,  $\epsilon>0$ and $x$ be given so that $B_\epsilon(x) \subset K^i(r)$. Set $B=B_{\epsilon/2}(x)$. For each $j \neq i$, $d(K^j(r),B)>\epsilon/2$, where $d(\cdot)$ is the Hausdorff metric,%
	\footnote{In this case it is actually a pseudo metric.}
	and continuity of $K^j$ implies that there exists a neighborhood $O_j$ of $r$ so that $d(K^j(r'),B)>\epsilon/4$ for all $r' \in O_j$.
	Let $O = \bigcap_{j \neq i} O_j$.
	Then, for any $r' \in O$, $B \cap cl( \bigcup_{j\neq i} K^j(r'))= \emptyset$.
	Since $cl(\bigcup_i K^i(r'))=X$, $B \subset cl(K^i(r'))$.
	But since $B$ is open, $B \subset int(cl(K^i(r')))=K^i(r')$ since $K^i(r')$ is regular open. \end{proof}

	For sufficiency, define a binary relation $S$ by $(a,b) S (c,d)$ if and only if $(a,c) \notin K^2(b,d)$. $S$ is clearly complete. 
	It is also transitive by S3.
	We show it has an open contour sets. Let $S^*$ be the strict part of $S$. If $(a,b) S^* (c,d)$, then  $x \in K^1(r)$ for $x=(a,c)$ and $r=(b,d)$.
	$K^1(r)$ is open by S0 so there exists $\epsilon>0$ so that $B_\epsilon(x) \subset K^1(r)$. By Lemma \ref{lem: usc}, 
	$x \in K^1(r')$ for all $r'$ in a neighborhood $O'$ of $r$.
	Conclude $(a',b') S^* (c',d')$ for all $(a',b'),(c',d')\in B_\epsilon(x) \times O'$.
	Standard results then show existence of a continuous function $\sigma$ so that $(a,b) S (c,d)$ if and only if $\sigma(a,b)\geq \sigma(c,d)$.
	$\sigma$ is symmetric by S2 and increasing in contrast by S1 and S4.
	Hence $x \in K^1(y)$ if and only if $\sigma(x_1,y_1)> \sigma(x_2,y_2)$, and by S2, $x \in K^2(y)$ if and only if $y' \in K^1(x')$ where $x',y'$ are the reflections of $x,y$. Hence, $x \in K^2(y)$ if and only if $\sigma(x_1,y_1) <  \sigma(x_2,y_2)$.
	
	Pick any $x,y>0$. We claim that $(x,y) \notin K^1(x,y) \cup K^2(x,y)$, and hence $\sigma(x,x)=\sigma(y,y)$. Observe that $(x+\epsilon,y) \in K^1(x,y)$  and $(x,y+\epsilon) \in K^2(x,y)$ by S4 for $\epsilon \neq 0$.  By S0, $K^1(x,y) , K^2(x,y)$ are open, so $(x,y) \notin K^1(x,y)$ and $(x,y)\notin  K^2(x,y)$. Conclude $\sigma$ is grounded.
	
	Pick any $a,b$. By S3, $\sigma(a,b)=\sigma(b,a)$ so $(a,b) \notin K^1(b,a)$ for any $a,b$. By S5, $(a+\epsilon,b) \notin K^1(b+\epsilon,a)$. 
	Then, $(b,a)S(a+\epsilon,b+\epsilon)$ so $\sigma(a,b)=\sigma(b,a)\geq \sigma(a+\epsilon,b+\epsilon)$. 
	Since $a,b$ were arbitrary, diminishing sensitivity holds.
	
	For necessity, verifying that S0-S5 hold are trivial, except that each $K^i(r)$ is regular open. To see this, pick $r$ and  $x \in int(cl(K^1(r)))$ (symmetric arguments hold for $K^2$). Suppose $x \gg r$ (the other cases follow by changing the signs). Then, there are $\epsilon_1, \epsilon_2$  such that $(x_1 - r_1)/2 > \epsilon_1 > 0,  \epsilon_2>0$ so that $\bar{x}=(x_1-\epsilon_1,x_2+\epsilon_2) \in cl(K^1(r))$. 
Since there exists $x' \in K^1(r)$ that is arbitrarily close to $\bar{x}$, we can find $x' \in K^1(r)$ so that $|x'_1-x_1|<\epsilon_1/2$ and $|x'_2-x_2|<\epsilon_2/2$.  In particular, $r_1<x'_1<x_1$ and $r_2<x_2<x'_2$. Then, $\sigma(x_1,r_1)>\sigma(x'_1,r_1)$ and  $\sigma(x'_2,r_2)>\sigma(x_2,r_2)$ since $\sigma$ is increasing in contrast. Moreover, $\sigma(x'_1,r_1)>\sigma(x'_2,r_2)$ since $x' \in K^1(r)$.  These inequalities imply $\sigma(x_1,r_1)>\sigma(x_2,r_2)$, hence $x \in K^1(r)$. Since $x$ was arbitrary, $int(cl(K^1(r)))\subset K^1(r)$. Clearly, $K^1(r) \subset int(cl(K^1(r)))$.

	Now we show the following are equivalent:
	
	(i) The  functions $K^1$ and $K^2$ satisfy S0, S1, and S6,
	
	(ii) There exists a salience function $\sigma$ s.t. $x \in K^k(r) \iff \sigma(x_k,r_k)>\sigma(x_{-k},r_{-k})$

	That (ii) implies (i) follows from the first part, and that S6 is implied by symmetry and homogeneity of degree zero.
	Now, we show (i) implies (ii). Set $\sigma(a,b)=\max\{a / b, b/a \}$.
		Clearly $\sigma$ is a salience function, and we show that $\sigma$ generates $K^1$ and $K^2$.
	Fix $r\in X$   and set $A = \{x: \sigma(x_1,r_1) > \sigma(x_2,r_2)\}$.
	We show $A=K^1(r)$.
	
	Claim $A \bigcap K^2(r) =\emptyset$.
	If not, pick $x \in A \bigcap K^2(r)$.
	$x \in A$ implies either (a) $x_1 / r_1 > x_2 / r_2$ and $x_1 / r_1 > r_2/x_2$ or (b) $r_1 / x_1 > x_2 / r_2$ and $r_1 / x_1 > r_2/x_2$.
	If (a) and $x_2 \leq  r_2$, then \[ x_1 / r_1 >  r_2 / x_2 \geq x_2 / r_2 \text{ implies }  
	x_1 > r_1 r_2 / x_2 \geq r_1,\] 
	so there exists $\lambda \in [0,1)$ such that $(\lambda x_1 +(1-\lambda) r_1 ,x_2)= ( r_1 r_2 / x_2,x_2)=x^\prime$.
	If (a) and $x_2 >  r_2$, then \[
	x_1 > r_1 x_2 / r_2 > r_1,\]
	so there exists $\lambda \in (0,1)$ such that $(\lambda x_1 +(1-\lambda) r_1 ,x_2)= ( r_1 x_2 / r_2,x_2)=x^\prime$.
	By S1 and $x \in K^2(r)$, $x^\prime \in K^2(r)$.
	However, we have either $x^\prime_1 x^\prime_2 =r_1 r_2$ or $x^\prime_1 / x^\prime_2 =r_1 / r_2$ so $x^\prime \notin K^2(r)$ by S6, a contradiction.
	A similar contradiction obtains if (b) holds.
	
	Now, since $A \bigcap K^2(r)=\emptyset$ and $K^1(r) \bigcup K^2(r)$ is dense, $A \subset cl(K^1(r))$.
	By S0, $K^1(r)=int (cl(K^1(r))$.
	Since $A$ is an open set contained in $cl(K^1(r))$, $A \subseteq K^1(r)$.
	Similarly, for $B =\{x: \sigma(x_1,r_1) < \sigma(x_2,r_2)\}$, $B \subseteq K^2(r)$.
	But \[(A\bigcup B)^c=\{x:x_1 x_2 = r_1 r_2\text{ or } x_1 / x_2 = r_1 /r_2\},\]
	and by S6, $(A\bigcup B)^c \bigcap K^k(r) = \emptyset$ for $k=1,2$.
	Thus $A =K^1(r)$ and $B =K^2(r)$, completing the proof.
	
	Finally, fix any HOD salience function $s$. Observe $s(a,b)>s(c,d)$ if and only if  $s(a/b,1)>s(c/d,1)$ by homogeneity if and only if  $s(\max(a/b,b/a),1)>s(\max(c/d,d/c),1)$ by symmetry if and only if  $\max(a/b,b/a)>\max(c/d,d/c)$ by ordering. Thus if one salience function generates the regions, every other salience function does as well.
\hfill\qedsymbol

\section{Proofs and Extras from Section \ref{sec: Choice Correspondence}}

\subsection{Axioms for $c$} 
\label{sec: axioms for c}
This subsection formally states the adaptations of the axioms for reference dependent preferences $\{\succsim_r\}_{r\in X}$ in terms of the choice correspondence $c$.
Interpretation is identical to that of those axioms.

\begin{ax1}[Category Cancellation*]
	For all $x_1, y_1, z_1,x_2,y_2,z_2\in \mathbb{R}_{++}$ and category $k$: if  $(x_1,z_2) \in c(S^1)$, $(z_1,y_2) \in S^1$,
	$(z_1,x_2) \in c(S^2)$, $ (y_1,z_2)\in S^2$, 
	$(x_1,x_2), (y_1, y_2)\in S^3$ and $S^i \subset K^k(A(S^i))$ for $i \in \{1,2,3\}$,  then $(x_1,x_2) \in c(S^3)$ whenever $ (y_1, y_2)\in c(S^3)$.
\end{ax1}

\begin{ax1}[Category Monotonicity*]
	For any $x,y \in X$: if $x \geq y$ and $x \neq y$, then 
	$(y,k) \not \succsim^{R} (x,k)$ for any category $k$. 
	
\end{ax1}

\begin{ax1}[Category Continuity*]
	and any $\epsilon>0$  so that $E \bigcap S \setminus c(S) = \emptyset$  where $E\equiv \bigcup_{x \in c(S)} B_\epsilon(x)$ there exists $\delta>0$ so that if $S' \in \mathcal{X}$,  $d(A(S'),A(S))<\delta$, and for any $y' \in S'$, there is $y \in S$ so that $y' \in B_\delta(y)$, then $c(S') \subset E$ whenever $S' \bigcap E \neq \emptyset$.
\end{ax1}

Define $\succsim^{R,k}$ by $x\succsim^{R,k} y$ if and only if $(x,k) \succsim^R (y,k)$.
Using this relation, we can define $\oplus^k$ for each category as we did with preference relations.

\begin{ax1}[Affine Across Categories*]
	For any $S^1,S^2,S^3\in \mathcal{X}$, $x^i \in K^j(A(S^i))$, $y^i \in K^k(A(S^i))$ for $i=1,2,3$, and any $\alpha \in (0,1)$ so that 
	$(x^3,j) \succsim^{R} (\alpha x^1 \oplus^j (1-\alpha) x^2,j) $ and 
	$(\alpha y^1 \oplus^k (1-\alpha) y^2,k) \succsim^{R} (y^3,k)$: \\
	if $x^1 \in c(S^1)$ and $x^2 \in c(S^2)$, then  $y^3\notin c(S^3) $. 
\end{ax1}
\begin{ax1}[Salient Dimension Overvalued*]
	\label{ax:SDOc}
	For $x,y \in S \bigcap S'$ with $x_k>y_k$ and $y_{-k}>x_{-k}$, if  $x,y \in K^k (A(S)) $, $x,y \in R_{-k} (A(S'))$, and $y \in c(S)$, then $x \notin c(S')$.
\end{ax1}

\begin{ax1}[Reference Interlocking*]
	For any $a,b,a',b',x',y',x,y\in  X$ with
	$x_{-i}=a_{-i}$, $y_{-i}=b_{-i}$,  $x'_{-i}=a'_{-i}$, $y'_{-i}=b'_{-i}$, $x_i=x'_i$, $y_i=y'_i$, $a_i=a'_i$, $b_i=b'_i$:\\
	if $x \sim^{R*}_k y$, $a \succsim^{R*}_k b$, and $x' \sim^{R*}_j y'$, then  it does not hold that $b' \succ^{R*}_j a'$.
\end{ax1}

\subsection{ Proof of Theorem \ref{thm: strong CTM choice}}
\label{pf: thm: strong CTM choice}
\begin{lemma}
	\label{lem: revealed k preference}
	Assume that Revealed Structure holds, and that $A$ is a generalized average. If Category-SARP, Category Monotonicity*, Category Cancellation*, and Category Continuity* hold, then for any category $k$ there exists a Category utility $U^k$ so that for any $x,y \in E^{R,k}$, $$(x,k)\succsim^{R} (y,k)\iff U^k(x) \geq U^k(y).$$
\end{lemma}
\begin{proof}
	Fix a category $i$ and pick any $x,y\in E^{R,i}$. Let $E^*=E^{R,i} \bigcap B_{d(x,y)+1}(x)$.
	As in proof of Lemma \ref{lem: trans closure complete}, there is a continuous path $\theta:[0,1]\rightarrow E^*$ so that $\theta(0)=x$ and $\theta(1)=y$ that crosses each $\succsim^{R,i}$ indifference curve at most once, and $Y=\theta^{-1}([0,1])$ is compact. We will show that for any $z\in Y$, there exists an open set $z\in B_z \subset E^*$ so that $\succsim^{R,i}$ is complete on $B_z$. If this is the case, we can mimic the rest of the proof of Lemma \ref{lem: trans closure complete} to show that either $x\succsim^{R,i} y$ or $y\succsim^{R,i} x$.
	
	By definition of $E^*$, for any $z\in E^*$, there exists $S\in \mathcal{X}$ with $A(S)=r$ so that $c(S)=z$.
	Since $K^i(r)$ is open, there exists $\epsilon_1>0$ so that $B_{2\epsilon_1}(z) \subset K^i(r)$.
	By Lemma \ref{lem: usc}, there exists $\epsilon_2>0$ so that $r' \in B_{\epsilon_2}(r)$ implies $B_{\epsilon_1}(z) \subset K^i(r')$.
	Pick $\zeta \in (0,\frac12)$ so that $B_\zeta(z)\cap S=z$.
	By Category Continuity*, there exists $\epsilon_3>0$ so that for any $S' \in \mathcal{X}$ with  $d(A(S'),A(S))<\epsilon_3$, for any $y' \in S'$, there is $y \in S$ so that $y' \in B_{\epsilon_3}(y)$, and $S' \bigcap B_{\zeta}(x) \neq \emptyset$, then $c(S') \subset B_{\zeta}(x)$.
	By Generalized Average, there exists $\epsilon_4>0$ so that $z' \in B_{\epsilon_4}(z)$ implies $d(A(S\setminus \{z\} \cup \{z'\}),A(S))<\min\{\epsilon_2,\epsilon_3\}/2$.
	Let $\epsilon^*=\min \{\epsilon_1,\epsilon_2,\epsilon_3,\epsilon_4,\zeta\}$.
	
	Pick any $x',y' \in B_{\epsilon^*/4}(z)$ and let $z_*=z-\frac12 \epsilon^*$.
	Set $S_0=S\setminus \{z\} \cup \{z_*\}$, noting $d(r,A(S_0))<\epsilon_2/2$.
	By Generalized Average,
	there exists $S^*$ with $\{x',y'\}\bigcup S_0 \subset S^*$
	so that $d(A(S^*),A(S_0))<\epsilon^*/2$ and $d(S_0,S^* \setminus [ \{x',y'\}\bigcup S_0 ])<(\epsilon^{*}/2)^2$.
	Since $d(A(S^*),r)\leq d(A(S^*),A(S_0))+d(A(S_0),r)<\epsilon_2$, $x',y' \in K^i(A(S^*))$.
	Since every member of $S^*$ is no more than $\epsilon^*$ away from a member of $S$, Category Continuity* implies that $c(S^*) \subset B_{\zeta}(z)$.
	CM* gives that either $x'\in c(S^*)$ or $y' \in c(S^*)$, so $x' \succsim^{R,i} y'$ or $y' \succsim^{R,i} x'$. 
	
	Continuity follows along the same lines as Lemma \ref{lem: cont trans closure}. CM* gives that it is also monotone, and Category Cancellation* that it is locally additive. Apply Theorem 2.2 of \cite{chateauneuf1993local} to get a globally additive representation $U^k$.
\end{proof}
	By Lemma \ref{lem: revealed k preference}, there exists a category utility $U^k$ for each category.
	Since $E^{R,k}$ is dense in $D^k$, we can extend $U^k$ to $D^k$ uniquely.
	By Generalized Average and Category Continuity*, for any $S\in X$ with $z \in [D^k \setminus E^{R_k} ] \cap S$, there is a $z' \in E^{R,k}$ arbitrarily close to $z$ so that $c(S)= c([S\setminus\{z\}]\cup \{z'\})$, so it is sufficient to establish a representation when all alternatives categorized as $k$ in $S$ belong $E^{R,k}$ for each $k$ and $S$.
	
	Fix two regions $k$ and $j$. By CAR, for any $x \in E^{R,k}$ there exists $x'\in E^{R,k}$, $y \in E^{R,j}$, and $S \in \mathcal{X}$ so that $x',y \in c(S)$ and $x \sim^{R,k} x'$. This implies there exists a strictly increasing function $H$ so that $V(x|r)=U^k(x)$ when $x \in K^k(r)$ and $V(x|r)=H(U^j(x))$ when $x \in K^j(r)$ represents choice (when $S \subset K^k \bigcup K^j$). This is well-defined and represents choice by Category SARP. By AAC*, $H$ is an affine function.
	The argument are readily seen to extend inductively to all regions, which complete the proof. 
	\hfill\qedsymbol
	
	\subsection{Proof of Lemma \ref{lem:BGS region}}
\label{pf: lem:BGS region}
Pick any $x \in X$ and set $S=\{x,x'\}$ where $x'=(\frac12 x_1,x_2)$. Then, $A(S)_2=x_2$ by strong generalized average, so both $x$ and $x'$ are $1$-salient by S4. By CM*, $x \in c(S)$, and so $x \in E^{R,1}$. $x$ was arbitrary, so  $X= E^{R,1}$. Similar for $K^2$.
\hfill\qedsymbol

\subsection{Proof of Proposition \ref{thm: BGS choice}}
By Lemma \ref{lem:BGS region}, the structure assumption is satisfied. By Theorem \ref{thm: BGS regions}, the category function is generated by a salience function. By Theorem \ref{thm: strong CTM choice}, $c$ conforms to Strong CTM. Mimicking the arguments of Theorem \ref{thm: RI}, Reference Interlocking implies $U^k(x)=w^k_1 u_1(x_1)+w^k_2 u_2(x_2)+\beta_k$.
The rest follows from the arguments that establish Proposition \ref{thm:BGS preference}.
\hfill\qedsymbol
\subsection{Proof of Proposition \ref{prop: revealing regions with c}}
\label{pf: prop: revealing regions with c}
 Suppose that $c$ has two BGS representations, $(U^1,U^2,\sigma)$ and $(U^{\prime 1},U^{\prime 2},\sigma')$.
We first show that $\sigma$ and $\sigma'$ categorize all alternatives $y \gg r$ the same for every $r$. Then, we use symmetry to show this implies they agree everywhere. Finally, we show that we can also directly reveal the category of $y$.

For contradiction, assume that $\sigma$ and $\sigma'$ disagree on the category of  $y$ for reference $r$ when $y\gg r$: $\sigma'(y_k,r_k)\leq\sigma'(y_{-k},r_{-k})$ and $\sigma(y_k,r_k)>\sigma(y_{-k},r_{-k})$. By continuity and increasing differences, we can take both inequalities to be strict by lowering $y_k$. Interchanging the role of the two representations if necessary,  there is no loss in assuming that  $U^k(r)\geq U^{-k}(r)$. 
By continuity, there exists $\epsilon>0$ so that if $d(r',r)<\epsilon$ and $d(y,y')<\epsilon$, then 
$$\sigma(y'_k,r'_k)>\sigma(y'_{-k},r'_{-k})\text{ and }\sigma'(y'_k,r'_k)<\sigma'(y'_{-k},r'_{-k}).$$ 
Pick $S$ so its convex hull is contained in $B_{\epsilon'}(r)$ and $U^k(y)>U^k(z),U^{-k}(z)$ for all $z \in B_{\epsilon'+\epsilon'^2}(r)$ for some $\epsilon'<\epsilon/2$; $\epsilon'$ exists by continuity of $U^k$ and $U^{-k}$.
Since $A$ is a strong generalized average, $d(A(S),r)<\epsilon/2$.
For any $y'$, generalized average implies there exists $S'$ so that $d(A(S'),A(S))<\epsilon'$, $d( S' \setminus\{y',y\},S)<\epsilon'^2 $, and $y,y' \in S'$. Label it $S(y')$ and note $d(A(S(y')),r)<\epsilon$.

Pick $y'$ with $d(y,y')<\epsilon$ so that $U^k(y')=U^k(y)$ and $y \neq y'$. As above, $U^{-k}(y)\neq U^{-k}(y')$. 
By Lemma \ref{lem:BGS region} and Theorem 2.2 of \cite{chateauneuf1993local}, $U^{\prime j}$ and $U^j$ agree up to an affine transformation for $j=1,2$, so $U^{\prime k}(y')=U^{\prime k}(y)$ and $U^{\prime -k}(y)\neq U^{\prime -k}(y')$ also.
Since $(U^{\prime 1},U^{\prime 2},\sigma')$ represents $c$, $y,y' \in K^k(A(S(y')))$ and  $c(S(y'))=\{y,y'\}$. 
However, $(U^{\prime 1},U^{\prime 2},\sigma')$ also represents $c$, so  $y,y' \in K^{-k}(A(S(y')))$. Hence, it is impossible that $c(S(y'))=\{y',y\}$; one has strictly higher utility than the other. This is a contradiction of both representing $c$, so conclude the categories coincide when $y\gg r$. 

We show that $\sigma$ and $\sigma'$ agree on the category of all alternatives whenever they agree whenever $y \gg r$. Pick any $x,y,a,b>0$. We show that $\sigma(x,a)>\sigma(y,a)$ if and only if  $x' \in K^1(r)$ for an appropriately chosen alternatives $x',r$ so that $x' \gg r$. This is impossible if $x=a$ and always true if $y=b$ and $x\neq a$. For any other values, it follows from symmetry of $\sigma$  that  $\sigma(x,a)>\sigma(y,b)$ if and only if either $(x,y) \in K^1 (a,b)$, $x>a$ and $y>b$; $(x,b) \in K^1 (a,y)$, $x>a$, and $b>y$;  $(a,y) \in K^1 (x,b)$, $x<a$, and $y>b$; or $(a,b) \in K^1 (x,y)$, $x<a$, and $b>y$.

We finally turn to directly revealing the salience of each alternative. As above, it suffices to consider $y \gg r$ and identify the categories of each alternative in $U(r)=\{x:x\gg r\}$. Again, pick $k$ so that $U^k(r) \geq U^{-k}(r)$ and define $S(y')$ as above.
If $y \in K^k(r)$, then there exists $\epsilon'>0$ so that $y,y' \in K^k(A(S(y'))$ whenever  $y' \in B_{\epsilon'}(y)$. It follows that $c(S(y'))=\{y,y'\}$ when $U^k(y)=U^k(y')$.
If $y \in K^{-k}(r)$, then there exists $\epsilon''>0$ so that $y,y' \in K^{-k}(A(S(y')))$ for all $y' \in B_{\epsilon''}(y)$.
For any such $y'$ with $U^k(y)=U^k(y')$ and $y'\neq y$, $c(S(y'))\neq \{y,y'\}$: $U^{-k} (y) \neq U^{-k}(y')$, so either $y'$ is not chosen, $y$ is not chosen, or both are not chosen, in which case one of the alternatives close to $r$ is chosen.
Since $K^k(r) \cup K^{-k}(r)$ is dense in $U(r)$, $K^k(r)\cap U(r)$ is the interior of the set of $y \gg r$ for which  there exists an $\epsilon'$ so that $c(S(y'))=\{y,y'\}$ when $U^k(y)=U^k(y')$ and $y' \in B_{\epsilon'}(y)$, and $K^{-k}(r)\cap U(r)$ is the interior of the set of $y \gg r$ for which  there exists an $\epsilon'$ so that $c(S(y'))\neq \{y,y'\}$ when $U^k(y)=U^k(y')$ and $y' \in B_{\epsilon'}(y)$.
\hfill\qedsymbol

\bibliographystyle{apa-good}

\bibliography{choice_new}
\end{document}